\documentclass[letterpaper, 10 pt, conference]{ieeeconf}  
\IEEEoverridecommandlockouts
\usepackage{balance}
\usepackage{color}
\usepackage{caption}
\usepackage{subcaption}
\usepackage{enumerate}
\usepackage{cite}
\usepackage{empheq}
\usepackage{mathrsfs}
\usepackage{multirow}
\usepackage{makecell}
\usepackage{graphicx}
\graphicspath{{./Figures/}}

\usepackage{caption}
\usepackage{mathtools}

\usepackage{tikz}
\usepackage{circuitikz}


\makeatletter
\pgfcircdeclarebipole{}{\ctikzvalof{bipoles/interr/height 2}}{spst}{\ctikzvalof{bipoles/interr/height}}{\ctikzvalof{bipoles/interr/width}}{

    \pgfsetlinewidth{\pgfkeysvalueof{/tikz/circuitikz/bipoles/thickness}\pgfstartlinewidth}

    \pgfpathmoveto{\pgfpoint{\pgf@circ@res@left}{0pt}}
    \pgfpathlineto{\pgfpoint{.6\pgf@circ@res@right}{\pgf@circ@res@up}}
    \pgfusepath{draw}   
}

\def\pgf@circ@spst@path#1{\pgf@circ@bipole@path{spst}{#1}}
\tikzset{switch/.style = {\circuitikzbasekey, /tikz/to path=\pgf@circ@spst@path, l=#1}}
\tikzset{spst/.style = {switch = #1}}
\makeatother

\makeatletter
\let\proof\@undefined                        
\let\endproof\@undefined                  
\makeatother
\usepackage{graphicx,amssymb,amstext,amsmath,amsthm}

\usepackage[bookmarks=true]{hyperref}
\usepackage{algorithm,algorithmicx,algpseudocode}
\algnewcommand{\algorithmicgoto}{\textbf{go to}}%
\algnewcommand{\Goto}[1]{\algorithmicgoto~\ref{#1}}%
\algnewcommand{\LineComment}[1]{\Statex \(\triangleright\) #1}
\algnewcommand{\LineCommentN}[1]{\Statex \hspace{1cm}\(\triangleright\) #1}

\usepackage{multirow}
\usepackage{stfloats}
\usepackage{cancel}

\newtheorem{prop}{Proposition} 
\newtheorem{cor}{Corollary}
\newtheorem{thm}{Theorem}
	\newtheorem{assumption}{Assumption}
\newtheorem{lem}{Lemma}
\newtheorem{defn}{Definition}

\newtheorem{problem}{Problem}

\usepackage{setspace}

\let\oldbibliography\thebibliography
\renewcommand{\thebibliography}[1]{%
  \oldbibliography{#1}%
}

\begin{document}

\title{\LARGE \bf Guaranteed State Estimation via Indirect Polytopic Set Computation for Nonlinear Discrete-Time Systems} 

\author{%
Mohammad Khajenejad, Fatima Shoaib, Sze Zheng Yong\\
\thanks{
M. Khajenejad, Fatima Shoaib and S.Z. Yong are with the School for Engineering of Matter, Transport and Energy, Arizona State University, Tempe, AZ, USA (e-mail: \{mkhajene, fshoaib, szyong\}@asu.edu).}
\thanks{This work is partially supported by NSF grant CNS-1932066.}
}

\maketitle
\thispagestyle{empty}
\pagestyle{empty}

\begin{abstract}
This paper proposes novel set-theoretic approaches for state estimation in bounded-error discrete-time nonlinear systems, subject to nonlinear observations/constraints. By transforming the polytopic sets that are characterized as zonotope bundles (ZB) and/or constrained zonotopes (CZ), from the state space to the space of the generators of ZB/CZ, we leverage a recent result on remainder-form mixed-monotone decomposition functions to   
compute the propagated set, i.e., a ZB/CZ that is guaranteed to enclose the set of the state trajectories of the considered system. 
Further, by applying the remainder-form decomposition functions to the nonlinear observation function, we derive the updated set, i.e., an enclosing ZB/CZ
of the intersection of the propagated set and the set of states that are compatible/consistent with the observations/constraints. Finally, we show that the mean value extension 
result in \cite{rego2020guaranteed} for computing propagated sets can also be extended to compute the updated set when the observation function is nonlinear.  
 \end{abstract}

\vspace{-0.15cm}
\section{Introduction} \vspace{-0.075cm}
State estimation is crucial in several research fields such as fault detection and isolation \cite{combastel2008fault}, localization problem \cite{jaulin2009nonlinear} and state-feedback control \cite{dahleh19871}. In such settings, Bayesian/stochastic estimation approaches such as particle or Kalman filtering can be applied if distributions/stochastic descriptions of uncertainties are known. However, in bounded-error settings where distribution-free set-valued uncertainties are considered, guaranteed sets of true values of state trajectories which are compatible/consistent with constraints/observations are desired. Obtaining the exact characterization of such sets that contain the evolution of the system states is very complicated and mostly intractable \cite{kieffer2004guaranteed}, hence developing set-theoretic approaches to tractably derive enclosures to such sets, while trying to make the enclosing sets as tight as possible, is critical. 

\emph{Literature review.} In the context of bounded-error settings, where dynamical systems are subject to distribution-free and bounded uncertainties, several seminal studies have proposed set-membership state estimation approaches for discrete-time constrained systems, to compute enclosing sets to all possible system trajectories \cite{chisci1996recursive,le2013zonotopic,rego2020guaranteed}. A well-known strategy, which is common to all these methods is finding an enclosing set to the image set of the dynamics vector field, i.e., \emph{propagation/prediction} step, as well as refining the obtained propagated set by finding an enclosure to its intersection with the set of states that are compatible/consistent with the observation/measurements, i.e., \emph{update} step. 

In case of linear systems with polytopic initial set, it is theoretically shown that tight (exact) enclosures can be obtained \cite{girard2008efficient}. However, even for linear systems, the computational complexity of polytopic propagation is extensive and grows dramatically with time \cite{shamma1999set}. Hence, simpler sets such as parallelotopes \cite{vicino1996sequential,chisci1996recursive}, ellipsoids \cite{khajenejad2019simultaneous,polyak2004ellipsoidal,khajenejadasimultaneous}, intervals \cite{zheng2016design,khajenejad2020simultaneousinterval1,wang2015interval,khajenejad2020simultaneousinterval2} or zonotopes \cite{le2013zonotopic,combastel2015merging} have been used to characterize the enclosures. However, structural limitations of these sets sometimes leads to conservative enclosures. To address this, the work in \cite{scott2016constrained} introduced \emph{constrained zonotpes} to ease some of the limitations imposed by zonotopes, while \emph{zonotope bundles} were proposed in \cite{althoff2011zonotope} to describe the intersection of zonotopes without explicit computations. 

Regarding nonlinear systems, obtaining efficient set-valued estimates is still very challenging, contrary to the linear case. A classical approach to tackle this problem has been to use interval arithmetic-based inclusion functions \cite{moore2009introduction} to propagate the current enclosing sets through the nonlinear dynamics and then to apply interval-based set inversion techniques (e.g., SIVIA) to find upper approximations for the set of compatibles states with the current measurements \cite{jaulin2009nonlinear,jaulin2016inner}. These approaches are computationally very efficient, but unfortunately, due to the nature of interval arithmetic, the resultant bounds are mostly conservative. 

Alternatively, given linear observation functions, zonotopic propagation methods have been developed in \cite{combastel2005state,alamo2005guaranteed,alamo2008set}, based on the first order Taylor expansion, the mean value extension or DC programming. However, significant errors are caused in update step due to the symmetry of zonotopes, even for linear measurements \cite{scott2016constrained}. More recently, the interesting work in \cite{rego2020guaranteed} proposed constrained zonotopic propagation and update algorithms for discrete-time nonlinear systems with linear observation functions, based on  mean value and first order Taylor extensions.           

 \emph{Contributions.} This paper  proposes novel methods for recursive state estimation (consisting of propagation and update steps) using polytopes (equivalently, constrained zonotopes or zonotope bundles) for nonlinear bounded-error discrete-time systems with nonlinear observation functions.  Leveraging remainder-form mixed-monotone decomposition functions \cite{khajenejad2021tightremainder} and following the standard propagation and update approach, this paper bridges the gap between constrained zonotope (CZ)/zonotope bundle (ZB)-based set-valued state estimation and nonlinear observation/constrained functions. In particular, for the propagation step, we transform the prior ZB/CZ's into the space of CZ/ZB generators, which are interval-valued, and further transform  
 the vector field into two components, one that is proven to attain tight image sets, as well as a linear remainder function, for which a family of  remainder-form mixed-monotone decomposition functions \cite{khajenejad2021tightremainder} can be obtained. Each of the decomposition functions produce enclosures of the state trajectories and thus, we can intersect them to obtain the desired propagated ZB/CZ enclosures. 

Moreover, we show that a similar idea, i.e., transformation from the state+uncertainty space to the space of generators of CZ/ZB's, can be used for the update step to find a family of enclosures to the \emph{generalized nonlinear intersection} of the propagated set with the set of states that is compatible with the observations, where the final enclosures are proven to be ZB/CZ's. Furthermore, we prove that the mean value extension approach used in \cite{rego2020guaranteed} to enclose a multiplication of an interval matrix to a constrained zonotope, {can also be} leveraged for the update step when the observation function is nonlinear. Finally, we compare our proposed approaches together and with the mean value extension-based approach in \cite{rego2020guaranteed}, {implementing it on} two examples, one with a linear and the other with a nonlinear observation function.
\section{Preliminaries}
In this section, we briefly introduce some of the main concepts that we use throughout the paper, as well as some important existing results that will be used for deriving our main results and for comparison.

{\emph{{Notation}.}} $\mathbb{N},\mathbb{N}_a, \mathbb{R}^n$ and $\mathbb{R}^{m \times n}$ denote the set of positive integers, the first $a$ positive integers, the $n$-dimensional Euclidean space and the space of $m$ by $n$ real matrices, respectively. For $\mathcal{Z},\mathcal{W} \subset \mathbb{R}^n, R \in \mathbb{R}^{m \times n}, \mathcal{Y} \subset \mathbb{R}^m$, and $\mu:\mathbb{R}^n \to \mathbb{R}^m$, $R\mathcal{Z}\triangleq \{Rz|z \in \mathcal{Z}\}, \mathcal{Z} \oplus \mathcal{W} \triangleq \{z+w | z \in \mathcal{Z},w \in \mathcal{W} \}$, $\mathcal{Z} \ominus \mathcal{W} \triangleq \{z-w | z \in \mathcal{Z},w \in \mathcal{W} \}$, $\mu(\mathcal{Z}) \triangleq \{\mu(z) | z \in \mathcal{Z}\}$ and $\mathcal{Z} \cup_{\mu} \mathcal{Y} \triangleq \{z \in \mathcal{Z} | \mu(z) \in \mathcal{Y} \}$ denote the linear mapping, Minkowski sum, set subtraction, general (nonlinear) mapping and generalized (nonlinear) intersection, respectively. Moreover, the transpose, 
Moore-Penrose pseudoinverse and $(i,j)$-th element 
of $R$ are given by $R^\top$, 
$R^\dagger$ and $R_{ij}$, 
 while 
its row support is $r=\textstyle{\mathrm{rowsupp}}(R) \in \mathbb{R}^m$, where $r_i=0$ if the $i$-th row of $R$ is zero and $r_i=1$ otherwise, $\forall i \in \mathbb{N}_m$. Furthermore, $\mathbb{B}^n_{\infty} \triangleq \{z \in \mathbb{R}^n | \|z\|_{\infty} \leq 1\}$ and $\mathbf{0}_n$ denote the $\infty$-norm hyperball and the zero vector in $\mathbb{R}^n$, respectively. For $z \in \mathbb{R}^n$, $\textstyle{\mathrm{diag}}(z)$ is a diagonal matrix in $\mathbb{R}^{n \times n}$, with its diagonal elements being the corresponding elements of $z$. $\langle \cdot,\cdot\rangle$ denotes the inner product operator.
\begin{defn}[Intervals, Polytopes, Constrained Zonotopes (CZ) and Zonotope Bundles (ZB)]\label{defn:CZ_ZB}
A set $\mathcal{Z} \subset \mathbb{R}^n$ is a(n) (i) interval, (ii) polytope, (iii) constrained zonotope (CZ), (iv) zonotope bundle (ZB), if
\begin{enumerate}[(i)]
\item $\exists \underline{z},\overline{z} \in \mathbb{R}^n$ such that $\mathcal{Z}=[\underline{z},\overline{z}] \triangleq \{z \in \mathbb{R}^n | \underline{z} \leq z \leq \overline{z}\}$. An interval matrix can be defined similarly, in an element-wise manner. 
\item $\exists A_p \in \mathbb{R}^{n_p \times n},b_p \in \mathbb{R}^{n_p}$ such that $\mathcal{Z}=\{A_p,b_p\}_P \triangleq \{ z \in \mathbb{R}^n |A_pz\leq b_p \}$. 
\item $\exists \tilde{G} \in \mathbb{R}^{n \times n_g}, \tilde{c}\in \mathbb{R}^{n},\tilde{A} \in \mathbb{R}^{n_c \times n_g},\tilde{b} \in \mathbb{R}^{n_c}$ such that $\mathcal{Z}=\{\tilde{G},\tilde{c},\tilde{A},\tilde{b}\}_{CZ} \triangleq \{ \tilde{G}\xi+\tilde{c} |\xi \in \mathbb{B}^{n_g},\tilde{A}\xi=\tilde{b} \}$. $n_g$ and $n_c$ are called the number of CZ's generators and constraints, respectively.
\item $\mathcal{Z}$ can be represented as an intersection of $S \in \mathbb{N}$ zonotopes, i.e., $\exists \{{G}_s \in \mathbb{R}^{n \times \hat{n}_s},{c}_s \in \mathbb{R}^n\}_{s=1}^S$ such that $\mathcal{Z}=\bigcap\limits_{s=1}^S\{{G}_s,{c}_s\}_{Z} \triangleq \bigcap\limits_{s=1}^S \{{G}_s\zeta+{c}_s |\zeta \in \mathbb{B}^{\hat{n}_g}\}$, with $\hat{n}_s,s=1,\dots,S$, being called the number of generators for each zonotope. 
\end{enumerate}
It is worth mentioning that using CORA2020 \cite{althoff2020cora}, if a set $\mathcal{Z}$ has each of the polytopic, CZ or ZB representations, it can be equivalently and tightly transformed into the other two representations. This is represented thorough this paper as:  $$\textstyle \mathcal{Z} =\{A_p,b_p\}_P \equiv \{\tilde{G},\tilde{c},\tilde{A},\tilde{b}\}_{CZ} \equiv \bigcap\limits_{s=1}^S\{{G}_s,{c}_s\}_{Z}.$$ 
\end{defn}
\begin{prop}\label{prop:interval_cen_rep}
Consider an interval vector $\mathbb{IZ} \triangleq [\underline{z},\overline{z}] \subset \mathbb{IR}^{n}$ and an interval matrix $\mathbb{J} \in \mathbb{IR}^{n \times m}$. Then, 
$\mathbb{IZ}$ and $\mathbb{J}$ can be equivalently represented as 
\begin{align}
\mathbb{IZ} &\triangleq [\underline{z},\overline{z}] \equiv \textstyle{\mathrm{mid}}(\mathbb{IZ})\oplus \frac{1}{2}\textstyle{\mathrm{diag}}(\textstyle{\mathrm{diam}}(\mathbb{IZ}))\mathbb{B}^n_{\infty},\label{eq:interv1} \\
\mathbb{J} &\triangleq [\underline{J},\overline{J}] \equiv \textstyle{\mathrm{mid}}(\mathbb{J})\oplus \mathbb{J}_{\Delta}\label{eq:interv2},
\end{align}
where for $q \in \{\mathbb{IZ},\mathbb{J}\}$, $\textstyle{\mathrm{mid}}(q) \triangleq \frac{1}{2}(\overline{z}+\underline{z})$, $\textstyle{\mathrm{diam}}(q) \triangleq (\overline{z}-\underline{z})$, and 
$\mathbb{J}_{\Delta} \in \mathbb{IR}^{n \times m}$ is an interval matrix that is defined as $[\mathbb{J}_{\Delta}]_{ij} \triangleq \frac{1}{2}[-\textstyle{\mathrm{diam}}(\mathbb{J})_{ij} \ \textstyle{\mathrm{diam}}(\mathbb{J})_{ij}], \forall i \in \mathbb{N}_{n},\forall j \in \mathbb{N}_{m}$. 
\end{prop}
\begin{proof}
To prove \eqref{eq:interv1}, consider $z \in \mathbb{IZ} \Leftrightarrow \underline{z} \leq z \leq \overline{z} \Leftrightarrow \underline{z}-\textstyle{\mathrm{mid}}(\mathbb{IZ}) \leq z-\textstyle{\mathrm{mid}}(\mathbb{IZ}) \leq \overline{z}-\textstyle{\mathrm{mid}}(\mathbb{IZ})\Leftrightarrow -\frac{1}{2}\textstyle{\mathrm{diam}}(\mathbb{IZ}) \leq z-\textstyle{\mathrm{mid}}(\mathbb{IZ}) \leq \frac{1}{2}\textstyle{\mathrm{diam}}(\mathbb{IZ})\Leftrightarrow \textstyle{\mathrm{mid}}(\mathbb{IZ})-\frac{1}{2}\textstyle{\mathrm{diam}}(\mathbb{IZ}) \leq z \leq \frac{1}{2}\textstyle{\mathrm{diam}}(\mathbb{IZ})+\textstyle{\mathrm{mid}}(\mathbb{IZ})\Leftrightarrow \exists \xi \in \mathbb{B}^n_{\infty}, s.t.\ z=\textstyle{\mathrm{mid}}(\mathbb{IZ})+\frac{1}{2}\textstyle{\mathrm{diag}}(\textstyle{\mathrm{diam}}(\mathbb{IZ}))\xi \Leftrightarrow z \in \textstyle{\mathrm{mid}}(\mathbb{IZ})\oplus \frac{1}{2}\textstyle{\mathrm{diag}}(\textstyle{\mathrm{diam}}(\mathbb{IZ}))\mathbb{B}^n_{\infty}$. The result in \eqref{eq:interv2} is a straightforward extension of \eqref{eq:interv1}. 
\end{proof}
\begin{prop}\cite[Theorem 1]{rego2020guaranteed}\label{prop:Rego_CZ_bounding}
Let $\mathcal{X}=\{G,c,A,b\}_{CZ} \subset \mathbb{R}^m$ be a constrained zonotope with $n_g$ generators and $n_c$ constraints, and $\mathbb{J} \in \mathbb{IR}^{n \times m}$ be an interval matrix. Consider the set $S=\mathbb{J}\mathcal{X} \triangleq \{Jx | J \in \mathbb{J}, x \in \mathcal{X} \} \subset \mathbb{R}^n$. Let $\overline{\mathcal{X}}=\{\overline{G},\overline{c}\}_Z$ be a zonotope satisfying $\mathcal{X} \subseteq \overline{\mathcal{X}}$ and $\overline{c} \in \mathbb{R}^{\overline{n}_g}$. Let $\mathbf{m} \in \mathbb{R}^n$ be an interval vector such that $\mathbf{m} \supset (\mathbb{J}-\textstyle{\mathrm{mid}}(\mathbb{J}))\overline{c}$ and $\textstyle{\mathrm{mid}}(\mathbf{m})=\mathbf{0}_n$. Let $P \in \mathbb{R}^{n \times n}$ be a diagonal matrix defined as follows. $\forall i=1,\dots,n$: 
\begin{align}\label{eq:Rego_P}
P_{ii}=\frac{1}{2}\textstyle{\mathrm{diam}}(\mathbf{m}_i)+\frac{1}{2}\sum_{j=1}^{\overline{n}_g} \sum_{k=1}^m \textstyle{\mathrm{diam}}(\mathbb{J}_{ik})|\overline{G}_{kj}|.
\end{align}
Then, $\mathcal{S} \subseteq \textstyle{\mathrm{mid}}(\mathbb{J})\mathcal{X} \oplus P\mathbb{B}^n_{\infty}$ 
\begin{align}\label{eq:CZ_bounding}
\quad \quad \quad =\{\begin{bmatrix}\textstyle{\mathrm{mid}}(\mathbb{J})G & P\end{bmatrix}\hspace{-.1cm}, \textstyle{\mathrm{mid}}(\mathbb{J})c, \begin{bmatrix} A & 0_{n_g \times n}\end{bmatrix}, b \}_{CZ}. 
\end{align}
\end{prop}
\begin{prop}[RRSR Propagation Approach] \cite[Theorem 2]{rego2020guaranteed}\label{prop:Rego_propagation}
Let $f:\mathbb{R}^n \times \mathbb{R}^{n_w} \to \mathbb{R}^n$ be continuously differentiable and $\nabla_xf$ denote the gradient of $f$ with respect to its first argument. Let $\mathcal{X}=\{G_x,c_x,A_x,b_x\}_{CZ} \subset \mathbb{R}^n$ and $\mathcal{W} \subset \mathbb{R}^{n_w}$ be constrained zonotopes. Choose any $h \in \mathcal{X}$. If $\mathcal{Z}$ is a constrained zonotope such that $f(h,\mathcal{W}) \subseteq \mathcal{Z}$ and $\mathbb{J} \in \mathbb{IR}^{n \times n}$ is an interval matrix satisfying $\nabla^\top_x f(\mathcal{X},\mathcal{W}) \subseteq \mathbb{J}$, then 
\begin{align}\label{eq:Rego_prediction}
f(\mathcal{X},\mathcal{W}) \subseteq \mathcal{Z} \oplus \textstyle{\mathrm{mid}}(\mathbb{J})(\mathcal{X}-h) \oplus \tilde{P}\mathbb{B}^n_{\infty}, 
\end{align}
where $\tilde{P}$ can be computed through \eqref{eq:Rego_P}, using $\mathbb{J}$ and an enclosing zonotope to $\mathcal{X}-h$. 
\end{prop}
\begin{defn}[Mixed-Monotone (One-Sided) Decomposition Functions For Discrete-Time Systems]\label{defn:dec_functions}
\cite[Definitions 3,4]{khajenejad2021tightremainder}
A mapping $f_d:\mathcal{Z} \times \mathcal{Z} \subset \mathbb{R}^{2n} \to \mathbb{R}^m$ is a discrete-time mixed-monotone decomposition function with respect to $f:\mathcal{Z} \subset \mathbb{R}^n \to \mathbb{R}^m$, over the set $\mathcal{Z}$, if it satisfies the following: 
$f_d(x,x)=f(x),x \geq x' \Rightarrow f_d(x,y) \geq f_d(x',y), y \geq y' \Rightarrow f_d(x,y) \leq f_d(x,y'), \forall x,y,x',y' \in \mathcal{Z}$. Further, if there exists two mixed-monotone mappings $\overline{f}_d,\underline{f}_d:\mathcal{Z} \times \mathcal{Z} \to \mathbb{R}^m$, such that for any $\underline{z},z,\overline{z} \in \mathcal{Z}$, the following holds: $\underline{z} \leq z \leq \overline{z} \Rightarrow \underline{f}_d(\underline{z},\overline{z}) \leq f(z) \leq \overline{f}_d(\overline{z},\underline{z})$, then $\overline{f}_d$ and $\underline{f}_d$ are called upper and lower decomposition functions for $f$ over $\mathcal{Z}$,  respectively.    
\end{defn}
It is trivial to see that $\forall x \in [\underline{x},\overline{x}], f_d(\underline{x},\overline{x}) \leq f(x) \leq f_d(\overline{x},\underline{x})$, where $f_d$ is a decomposition function of $f$.
\begin{prop}[Tight and Tractable Remainder-Form Upper and Lower Decomposition Decomposition Functions]\cite[Theorems 1,2,3 ]{khajenejad2021tightremainder}\label{prop:Remainder_DF}
Consider a locally Lipschitz vector field ${f}_i:\mathbb{IZ} \triangleq [\underline{z},\overline{z}] \subseteq \mathbb{IR}^{n_{z}} \to \mathbb{R}$. Let 
$\mathbb{N}_{n_{z}} \triangleq \{1,\dots,n_{z}\}$ and $\overline{{J}}^{\tilde{f}}_i,\underline{{J}}^{\tilde{f}}_i \in \mathbb{R}^{n_{z}}$ denote the upper and lower bounds for the Jacobian matrix (vector) of ${f}_i$ over $\mathbb{IZ}$. Suppose that Assumption \ref{assumption:mix-lip} in Section \ref{sec:formulation} holds. Then, $f_i(\cdot)$ admits a family of mixed-monotone remainder-form decomposition functions denoted as $\{{f}_{d,i}(z,\hat{z};m,h(\cdot))\}_{\mathbf{m} \in \mathbf{M}_i,h(\cdot) \in \mathcal{H}_{\mathbf{M}^c_i}}$, that is parametrized by a set of supporting vectors $\mathbf{m} \in \mathbf{M}^c_i$ 
\begin{align}
 \nonumber \mathbf{m} \in \mathbf{M}^c_i \triangleq \{&\mathbf{m} \in \mathbb{R}^ {n_z} | \mathbf{m}_j = \min(\underline{{J}}^f_{ij},{0}) \ \lor \\
 &\mathbf{m}_j = \max(\overline{{J}}^f_{i,j},{0}) 
  ,
 \forall j \in \mathbb{N}_{n_z} \}.\label{eq:Mj}
  \end{align}
and a locally Lipschitz remainder function $h(\cdot) \in \mathcal{H}_{M^c_i}$, where 
{\begin{align}\label{eq:decomp1} 
\hspace{-.3cm}f_{d,i}(z,\hat{z};\mathbf{m},h(\cdot))\hspace{-.1cm}=\hspace{-.1cm} h(\zeta_{\mathbf{m}}(\hat{z},{z}))\hspace{-.1cm}+\hspace{-.1cm}f_i(\zeta_{\mathbf{m}}(z,\hat{z}))\hspace{-.05cm}) \hspace{-.1cm}-\hspace{-.1cm}h_i(\zeta_{\mathbf{m}}(z,\hat{z})),
\end{align}}
$\zeta_{\mathbf{m}}(z, \hat{z})=[\zeta_{\mathbf{m},1}(z,\hat{z}),\dots,\zeta_{\mathbf{m},n_z}(z,\hat{z})]^\top$, $\forall j \in \mathbb{N}_{n_z}$:
\begin{align}\label{eq:corners}
{\zeta}_{\mathbf{m},j}(z,{\hat{z}})\hspace{-.1cm}=\hspace{-.1cm}\begin{cases} \hat{z}_j, \ \text{if} \ \mathbf{m}_j \hspace{-.1cm}=\hspace{-.1cm} \max(\overline{{J}}^f_{i,j},{0}) \\
{z}_j, \ \text{if} \ \mathbf{m}_j \hspace{-.1cm}=\hspace{-.1cm} \min(\underline{{J}}^f_{i,j},{0}) 
   \end{cases},
 \end{align} 
 and $\mathcal{H}_{\mathbb{M}_i} \triangleq \{h:\mathbb{IZ} \to \mathbb{R} | [\underline{J}^h(z),\overline{J}^h_C(z)] \subseteq \mathbb{M}_i,\forall z \in \mathbb{IZ} \}$.
 Moreover, search for the tightest mixed-monotone upper and lower  remainder-form decomposition functions in the form of \eqref{eq:decomp1} can be equivalently restricted to the set of ``linear remainders", parametrized by $\mathbf{m} \in \mathbf{M}^c_i$, i.e., 
 linear remainders $\{{h}(\cdot)\}_{\mathbf{m} \in \mathbf{M}^c_i}=\{{\langle}\mathbf{m}^i,\cdot{\rangle}\}_{\mathbf{m} \in \mathbf{M}^c_i}$. 
\end{prop}
\begin{cor}\label{cor:supproting_vectors}
Consider a locally Lipschitz mapping $\tilde{f}(\cdot):{\mathbb{I}\Xi} \triangleq [\underline{\xi},\overline{\xi}] \subseteq \mathbb{IR}^{n_{\xi}} \to \mathbb{R}^{n_x}$ that satisfies the assumptions in Proposition \ref{prop:Remainder_DF}. Let us define: $\mathbb{N}_{n_{x}} \triangleq \{1,\dots,n_{x}\}$ and
\begin{align}\label{eq:JSS_maker}
  \mathbf{H}_{\tilde{f}}  \triangleq  \{H  \in  \mathbb{R}^{ n_{x} \times n_{\xi}} | H^\top_{i,:}\in \mathbf{M}^c_i, \forall i \in \mathbb{N}_{n_{x}}  \}, 
\end{align} 
where $\mathbf{M}^c_i$ is defined in \eqref{eq:Mj}. Then, $\forall \xi \in \mathbb{I}\Xi,\forall H \in \mathbf{H}^{\tilde{f}}, \tilde{g}^H(\xi)\triangleq \tilde{f}(\xi)-H\xi $ is proven to be a Jacobian sign-stable (JSS) function, i.e., {$\forall i \in \mathbb{N}_{n_x}, \forall j \in \mathbb{N}_{n_z}, J^{H}_{ij}(\xi)\triangleq \frac{\partial f_i}{\partial \xi_j}(\xi)\geq 0, \forall \xi \in  \mathbb{I}\Xi$ or $J^{H}_{ij}(\xi)\triangleq \frac{\partial \tilde{g}^H_i}{\partial \xi_j}(\xi)\leq 0, \forall \xi \in  \mathbb{I}{\Xi}$.}. Consequently, $\tilde{g}^H(\cdot)$ can be tightly bounded in each dimension $i \in \mathbb{N}_{n_x}$ by remainder-form decomposition functions $\tilde{g}_{d,i}(\cdot,\cdot;H^\top_{i,:},\langle H^\top_{i,:},\cdot \rangle)$, constructed using \eqref{eq:decomp1}--\eqref{eq:corners}, as follows: 
\begin{align*}
 \tilde{g}_{d,i}(\underline{\xi},\overline{\xi};H^\top_{i,:},\langle H^\top_{i,:},\cdot \rangle)  \leq \tilde{g}_i(\xi) \leq  \tilde{g}_{d,i}(\overline{\xi},\underline{\xi};H^\top_{i,:},\langle H^\top_{i,:},\cdot \rangle),  
\end{align*} 
where, by \cite[Lemma 3]{khajenejad2021tightremainder} and defining $m \triangleq H^\top_{i,:}$, we obtain $\tilde{g}_{d,i}(\overline{\xi},\underline{\xi};{m},\langle {m},\cdot\rangle)=\tilde{f}_i(\zeta^+_m)+{m}^\top(\zeta^-_{{m}}-\zeta^+_{{m}})$, $\tilde{g}_{d,i}(\underline{\xi},\overline{\xi};{m}^\top,\langle {m},\cdot\rangle)=\tilde{f}_i(\zeta^-_{{m}})+{m}^\top(\zeta^+_{{m}}-\zeta^-_{{m}})$, $\zeta^+_m \triangleq \zeta_m({\xi},\underline{\xi})$, $\zeta^-_m \triangleq \zeta_m(\underline{\xi},\overline{\xi})$, with $\zeta_m(\cdot,\cdot)$ given in \eqref{eq:corners}. 
\begin{proof}
The proof follows the lines of the proof of \cite[Lemma 1, Proposition 10 and Corollary 2]{khajenejad2021tightremainder}.
\end{proof}
\end{cor}
\section{Problem Formulation} \label{sec:Problem} \label{sec:formulation}
\vspace{-0.1cm}
\noindent\textbf{\emph{System Assumptions.}} 
Consider the following bounded-error nonlinear constrained discrete-time system: 
\begin{align} \label{eq:system}
\begin{array}{ll}
x_{k+1}&=\hat{f}(x_k,w_k,u_k)=f(z_k), \\
\hat{\mu}(x_k, u_k)&=\mu(x_k) \in \mathcal{Y}_k, \ x_0 \in \hat{\mathcal{X}}_0,w_k \in \mathcal{W}_k,  \end{array}
\end{align}
where $z_k \triangleq [x^\top_k   w^\top_k]^\top$, $x_k \in \mathbb{R}^{n_x}$ is the state vector, $w_k \in \mathcal{W}_k \subset \mathbb{R}^{n_w}$ is the augmentation of all the exogenous uncertain inputs, e.g., bounded process disturbance/noise and internal uncertainties such as uncertain parameters and $u_k \in \mathcal{U}_k \subseteq \mathbb{R}^{n_u}$ is the \emph{known} input signal. Furthermore, 
$f :\mathbb{R}^{n_z} \to \mathbb{R}^{n_x}$ (with $n_z \triangleq n_x+n_w$) and $\mu : \mathbb{R}^{n_x}  \to  \mathbb{R}^{n_{\mu}}$ 
are nonlinear state vector field and observation/constraint mapping, respectively, which are well-defined, given $\hat{f}(\cdot,\cdot)$ and $\hat{\mu}(\cdot,\cdot)$, as well as the fact that $u_k$ is known. {Note that the mapping $\mu(\cdot)$ along with the set $\mathcal{Y}_k$ characterize all the existing and/or a prior known or even manufactured/redundant constraints over the states, observations and measurement noise signals or uncertain parameters at time step $k$.}.

The initial state estimate $\hat{\mathcal{X}}_0$ is assumed to be a known set a prior.  
Moreover, we assume the following assumptions.
\begin{assumption}\label{ass:polytopic_uncertainty}
The sets $\hat{\mathcal{X}}_0$, as well as $\mathcal{W}_k,\mathcal{U}_k, \forall k \geq 0$ are prior known polytopes, or equivalently constrained zonotopes or zonotope bundles (cf. Definition \ref{defn:CZ_ZB}). 
\end{assumption}
\begin{assumption}\label{assumption:mix-lip}
The nonlinear vector fields $f(\cdot)$ and $\mu(\cdot)$ are locally Lipschitz on their domains. Consequently, they are differentiable and have bounded Jacobian matrix elements, almost everywhere. We further assume that given any $\mathcal{Z} \subset \mathbb{R}^{n_z}$ and $\mathcal{X} \subset \mathbb{R}^{n_x}$, some upper and lower bounds for all elements of Jacobian matrices for $f(\cdot)$ and $\mu(\cdot)$ over $\mathcal{Z}$ and $\mathcal{X}$ are available or can be computed. In other words, $\exists \underline{{J}}^f,\overline{{J}}^f \in \mathbb{R}^{n_x \times n_z},\underline{{J}}^\mu,\overline{{J}}^\mu \in \mathbb{R}^{n_{\mu} \times n_x}$, such that: $\underline{{J}}^f \leq {{J}}^f(z) \leq \overline{{J}}^f,\underline{{J}}^{\mu} \leq {{J}}^{\mu}(x) \leq \overline{{J}}^{\mu},\forall z \in {\mathcal{Z}},\forall x \in \mathcal{X}$, where ${J}^f(z)$ and ${J}^{\mu}(x)$ denote the Jacobian matrices of the mappings $f(\cdot)$ and $\mu(\cdot)$ at the points $z$ and $x$, respectively. 
\end{assumption}
In this paper, we aim to propose novel set-membership approaches for obtaining set-valued state estimates for bounded-error nonlinear systems in the form of \eqref{eq:system}. More formally, proceeding the well-known two-step i) Propagation (prediction), ii) update (refinement) approach and given the initial uncertainty set $\hat{\mathcal{X}}_0$, we seek to solve the following corresponding problems at each time step $k \in \mathbb{N}$ (where $\mathcal{Z}_k \triangleq \mathcal{X}^u_k \times \mathcal{U}_k, \mathcal{X}^p_0 \triangleq \hat{\mathcal{X}}_0$):
\begin{problem}[Propagation]\label{prob:propagation}
Find the ``propagated set" $\mathcal{X}^p_k$ that satisfies
\begin{align}\label{eq:propagation}
f(\mathcal{Z}_{k-1})\hspace{-.1cm}\triangleq \hspace{-.1cm}\{\hat{f}(x,w,u_{k-1})|x \hspace{-.1cm} \in \hspace{-.1cm}\mathcal{X}^u_{k-1}, w \in \mathcal{W}_k \} \subseteq \mathcal{X}^p_k.
\end{align}
\end{problem}
\begin{problem}[Update]\label{prob:update}
Find the updated set $\mathcal{X}^u_k$ that satisfies 
\begin{align}\label{eq:update}
\mathcal{X}^p_k \cap_{\mu} \mathcal{Y}_k \triangleq \{x \in \mathcal{X}^p_k | \mu(x) \in \mathcal{Y}_k \}  \subseteq \mathcal{X}^u_k.
\end{align}
\end{problem}
\section{Main Results} \label{sec:main}
To address the above problems, we propose set-theoretic approaches which are based on our recently developed tight remainder-form mixed-monotone decomposition functions \cite{khajenejad2021tightremainder} (also cf. Proposition \ref{prop:Remainder_DF} and Corollary \ref{cor:supproting_vectors}), to compute supersets of i) $f(\mathcal{Z})$, and ii) $\mathcal{Z}_f \cap_{\mu} \mathcal{Z}_{\mu}$, given locally Lipschitz functions $f,\mu$ that satisfy Assumption \ref{assumption:mix-lip} and arbitrary sets $\mathcal{Z},\mathcal{Z}_f,\mathcal{Z}_{\mu}$ that satisfy Assumption \ref{ass:polytopic_uncertainty}. Then, armed with the aforementioned proposed set-membership approaches, we sequentially apply the one-step operations in \eqref{eq:propagation} and \eqref{eq:update} 
at each time step $k \in \mathbb{N}$ to system \eqref{eq:system}, with the purpose of  computing i) the propagated set $\mathcal{X}^p_k$ that satisfies the inclusion in  \eqref{eq:propagation}, 
and ii) the updated set $\mathcal{X}^u_k$ that satisfies the inclusion in \eqref{eq:update}, respectively. 
\subsection{Decomposition-Based ZB/CZ propagation} \label{sec:propagation}
In this section, we address Problem \ref{prob:propagation}, assuming that the initial set is a zonotope bundle (Lemma \ref{lem:ZB_propagation}) or a constraint zonotope (Lemma \ref{lem:CZ_propagation}). The main idea is to ``lift up" the initial ZB/CZ's from the $z$-space, i.e., the space of augmented state $x$ and process uncertainty $w$, to intervals in the $\xi$-space, i.e., the space of ZB/CZ generators. Then, based on our recent results in \cite{khajenejad2021tightremainder}, we decompose the transformed vector fields in the $\xi$-space into two components, a Jacobian sign stable (JSS) and a linear remainder mapping (cf. Corollary \ref{cor:supproting_vectors}). Finally, we apply our recently developed family of mixed-monotone remainder-form decomposition functions to find enclosures to the JSS components, with interval domains, which are proven to be tight by Corollary \ref{cor:supproting_vectors}. Using these tight obtained bounds and thanks to linearity of the remainders, we show that by augmenting and intersecting all the obtained enclosures, the resultant set is a ZB/CZ. We formally summarize our proposed Decomposition-Based ZB/CZ approach through the following Lemmas \ref{lem:ZB_propagation} and \ref{lem:CZ_propagation}.         
\begin{lem}[Decomposition-Based ZB Propagation]\label{lem:ZB_propagation}
Let $f:\mathcal{Z}\subset \mathbb{R}^{n_z}  \to \mathbb{R}^{n_x}$ satisfies Assumption \ref{assumption:mix-lip}. Let $\mathcal{Z}$ be a ZB in $\mathbb{R}^{n_z}$, i.e., $\mathcal{Z}=\bigcap\limits_{s=1}^{S} \{G_{s},c_{s}\}_Z$, and $\forall s \in \mathbb{N}_{S} \triangleq \{1,\dots,S\}$, $n_s$ be the number of generators of the corresponding zonotope. Then, the following set inclusion holds:
\begin{align}
f(\mathcal{Z}) \subseteq \mathcal{ZB}_f \triangleq\bigcap\limits_{s=1}^{S} \bigcap\limits_{H_s \in \mathbf{H}_{\tilde{f}_s}}\{G^{H_s}_{s},c^{H_s}_{s}\}_Z, \label{eq:ZB_prediction}
\end{align}
where $G^{H_s}_{s} \triangleq [H_s \  \frac{1}{2}\textstyle{\mathrm{diag}}(\overline{g}^{H_s}_{s}-\underline{g}^{H_s}_{s})], \ c^{H_s}_{s} \triangleq  \frac{1}{2}(\overline{g}^{H_s}_{s}+\underline{g}^{H_s}_{s})$, 
\begin{align}
& \overline{g}^{H_s}_{s,i} \triangleq \label{eq:ov_g} g^{s}_{i,d}(\mathbf{1}_{n_{s}},-\mathbf{1}_{n_{s}};{H_s^\top}_{(i,:)},\langle {H_s^\top}_{(i,:)},\cdot\rangle),\\
&\underline{g}^{H_s}_{s,i} \triangleq \label{eq:und_g} g^{s}_{i,d}(-\mathbf{1}_{n_{s}},\mathbf{1}_{n_{s}};{H_s^\top}_{(i,:)},\langle {H_s^\top}_{(i,:)},\cdot\rangle),
\end{align}
$g^{s}_{i,d}(\cdot,\cdot;{H_s^\top}_{(i,:)},\langle {H_s^\top}_{(i,:)},\cdot\rangle)$ is the tight mixed-monotone decomposition function (cf. Proposition \ref{prop:Remainder_DF}) for the JSS mapping $g^{H_s}_{s,i}(\xi) \triangleq \tilde{f}_{s,i}(\xi)-\langle {H_s^\top}_{(i,:)}, \xi \rangle : \mathbb{B}^{n_s}_{\infty} \to \mathbb{R}^{n_x}$, $\mathbf{H}_{\tilde{f}_s}$ is defined in Corollary \ref{cor:supproting_vectors} (with  the  corresponding  function  being $\tilde{f}_s$) 
and $\tilde{f}_{s}(\xi) \triangleq f({c}_s+{G}_s\xi)$.
\end{lem}
\begin{proof}
To show \eqref{eq:ZB_prediction}, $\forall s \in \mathbb{N}_S$, consider the zonotope $\mathcal{Z}_s \triangleq \{G_s,c_s\}_Z \triangleq \{z=G_s\xi+c_s|\xi \in \mathbb{B}^{n_s}_{\infty}\}$ and let us define $\tilde{f}_s(\xi):\mathbb{B}^{n_s}_{\infty} \to \mathbb{R}^{n_x} \triangleq f(G_s\xi+c_s)$ that implies
\begin{align}\label{eq:set_equiv}
f(\mathcal{Z}_s) \subseteq \tilde{f}_s(\mathbb{B}^{n_s}_{\infty}), \forall s \in \mathbb{N}_S .
\end{align}
On the other hand, note that by Corollary \ref{cor:supproting_vectors}, $\forall H_s \in \mathbf{H}_{\tilde{f}_s}$, $\tilde{f}_s(\cdot)$ can be decomposed as 
\begin{align}\label{eq:decomposed}
\tilde{f}_s(\xi)\hspace{-.1cm} =\hspace{-.1cm}g^{H_s}_s(\xi)\hspace{-.1cm}+\hspace{-.1cm}H_s\xi, \forall s \in \mathbb{N}_S, \forall \xi \in \mathbb{B}^{n_s}_{\infty},\forall H_s \in \mathbf{H}_{\tilde{f}_s} 
\end{align}
where $g^{H_s}_s(\xi)$ is a JSS function in $\mathbb{B}^{n_s}_{\infty}$ and $\mathbf{H}_{\tilde{f}_s}$ can be computed from \eqref{eq:JSS_maker}, with the corresponding function being $\tilde{f}_s$. Now \eqref{eq:set_equiv} and \eqref{eq:decomposed} together imply:
\begin{align}\label{eq:sum_decomp}
f(\mathcal{Z}_s) \subseteq g^{H_s}_s(\mathbb{B}^{n_s}_{\infty})\oplus H_s\mathbb{B}^{n_s}_{\infty}, \forall s \in \mathbb{N}_S, \forall H_s \in \mathbf{H}_{\tilde{f}_s}. 
\end{align}
Again, it follows from Corollary \ref{cor:supproting_vectors} and the fact that $g^{H_s}_s(\xi)$ is a JSS function that in each dimension $i \in \mathbb{N}_{n_x}$, $g^{H_s}_{s,i}(\xi)$ can be tightly bounded as $\underline {g}^{H_s}_{s,i}\leq g^{H_s}_{s,i}(\xi) \leq \overline {g}^{H_s}_{s,i}, \forall \xi \in \mathbb{B}^{n_s}_{\infty}, \forall H_s \in \mathbf{H}_{\tilde{f}_s}  $, with $\overline {g}^{H_s}_{s,i},\underline {g}^{H_s}_{s,i}$ given in \eqref{eq:ov_g} and \eqref{eq:und_g}, respectively. Augmenting all these $\mathbb{N}_{nx}$ one-dimensional inequalities yields the following set inclusion for all $s \in \mathbb{N}_{S}$ and all $H_s \in \mathbf{H}_{\tilde{f}_s} $:
${g}^{H_s}_{s}(\mathbb{B}^{n_s}_{\infty}) \subseteq [\underline {g}^{H_s}_{s},\overline {g}^{H_s}_{s}]=\frac{1}{2}((\underline {g}^{H_s}_{s}+\overline {g}^{H_s}_{s})\oplus \textstyle{\mathrm{diag}}(\overline {g}^{H_s}_{s}-\underline {g}^{H_s}_{s})\mathbb{B}^{n_x}_{\infty})$,  
where the last equality follows from Proposition \ref{prop:interval_cen_rep}. This, \eqref{eq:sum_decomp} and the fact that the inclusion in \eqref{eq:sum_decomp} holds for all $s \in \mathbb{N}_{S}$ and all $H_s \in \mathbf{H}_{\tilde{f}_s} $ and hence for the intersection of all of them, return the result in \eqref{eq:ZB_prediction}.
\end{proof}\vspace{-.3cm}
\begin{lem}[Decomposition-Based CZ Propagation]\label{lem:CZ_propagation}
Let $f:\mathcal{Z}\subset \mathbb{R}^{n_z}  \to \mathbb{R}^{n_x}$ satisfies Assumption \ref{assumption:mix-lip}. Let $\mathcal{Z}$ be a CZ in $\mathbb{R}^{n_z}$, i.e., $\mathcal{Z}= \{\tilde{G},\tilde{c},\tilde{A},\tilde{b}\}_{CZ}$, and $n_g$ be the number of generators of $\mathcal{Z}$. Then, the following set inclusion holds: 
\begin{align} 
f(\mathcal{Z}) \subseteq \mathcal{CZ}_f \triangleq \bigcap\limits_{H \in \mathbf{H}_{\tilde{f}}} \{\tilde{G}^H,\tilde{c}^H,\mathbb{A},\tilde{b}\}_{CZ},\label{eq:CZ_prediction} 
\end{align} 
where $\tilde{G}^H \triangleq [H \ \  \frac{1}{2}\textstyle{\mathrm{diag}}(\overline{g}^H-\underline{g}^H)], \mathbb{A} \triangleq [\tilde{A} \ \ 0_{n_g \times n_x}]$, 
\begin{align}
 \overline{g}^H_{i} &\triangleq \label{eq:ov_g_cz} \tilde{g}_{i,d}(\overline{\mathbf{l}}_{n_{g}},\underline{\mathbf{l}}_{n_{g}};H^\top_{i,:},\langle H^\top_{i,:},\cdot\rangle),\tilde{c}^H \hspace{-.1cm}\triangleq \hspace{-.1cm} \frac{1}{2}(\overline{g}^H+\underline{g}^H),\\
\underline{g}^H_{i} &\triangleq \label{eq:und_g_cz} \tilde{g}_{i,d}(\underline{\mathbf{l}}_{n_{g}},\overline{\mathbf{l}}_{n_{g}};H^\top_{i,:},\langle H^\top_{i,:},\cdot\rangle), \\
\nonumber \overline{\mathbf{l}}_{n_{g}} & \triangleq \min(\mathbf{1}_{n_g},\tilde{A}^\dagger\tilde{b}+\kappa \mathbf{r}_{n_g}),\underline{\mathbf{l}}_{n_{g}} \hspace{-.1cm} \triangleq \hspace{-.1cm} \max(-\mathbf{1}_{n_g},\tilde{A}^\dagger\tilde{b}\hspace{-.1cm}-\hspace{-.1cm}\kappa \mathbf{r}_{n_g})
\end{align}
$\tilde{g}_{i,d}(\cdot,\cdot;H^\top_{i,:},\langle H^\top_{i,:},\cdot\rangle)$ is the tight mixed-monotone decomposition function (cf. Proposition \ref{prop:Remainder_DF}) for the JSS mapping $\tilde{g}_{i}(\xi) \triangleq \tilde{f}_{i}(\xi)-\langle H^\top_{i,:}, \xi \rangle : \mathbb{B}^{n_g}_{\infty} \to \mathbb{R}^{n_x}$, $\mathbf{H}_{\tilde{f}}$ is defined in Corollary \ref{cor:supproting_vectors} and $\tilde{f}(\xi) \triangleq f(\tilde{c}+\tilde{G}\xi)$, $\mathbf{r}_{n_g} \triangleq \textstyle{\mathrm{rowsupp}}(I_{n_g}-\tilde{A}^\dagger \tilde{A})$ and $\kappa$ is a very large positive real number (infinity). 
\end{lem}
\begin{proof}
To prove the inclusion in \eqref{eq:CZ_prediction}, consider the constrained zonotope representation of the set $\mathcal{Z}$, i.e., $\mathcal{Z} \triangleq \{\tilde{G},\tilde{c},\tilde{A},\tilde{b}\}_{CZ} \triangleq \{z=\tilde{G}\xi+c|\xi \in \mathbb{B}^{n_g}_{\infty}, \tilde{A}\xi=\tilde{b}\}$. Using similar notation as in the proof of Lemma \ref{lem:ZB_propagation}, let us define $\tilde{f}(\xi):\mathbb{B}^{n_g}_{\infty} \to \mathbb{R}^{n_x} \triangleq f(\tilde{G}\xi+\tilde{c})$ that consequently return
\begin{align}\label{eq:set_equiv_cz}
f({\mathcal{Z}}) \subseteq \{\tilde{f}(\xi)|\xi \in \mathbb{B}^{n_g}_{\infty},\tilde{A}\xi=\tilde{b} \} .
\end{align}
Note that by \cite[Theorem 2]{james1978generalised}, $\tilde{A}\xi=\tilde{b} \Rightarrow \xi \in \mathbb{I}{\Xi} \triangleq [\tilde{A}^\dagger\tilde{b}-\kappa \mathbf{r}_{n_g},\tilde{A}^\dagger\tilde{b}-\kappa \mathbf{r}_{n_g}]$, where $\mathbf{r}_{n_g} \triangleq \textstyle{\mathrm{rowsupp}}(I_{n_g}-\tilde{A}^\dagger \tilde{A})$ and $\kappa$ is a very large positive real number. This, in addition to the fact that $\xi \in \mathbb{B}^{n_g}_{\infty}$ (cf. \eqref{eq:set_equiv_cz}), imply that $\xi \in \mathbb{I}\tilde{\Xi}\triangleq \mathbb{I}{\Xi} \cap \mathbb{B}^{n_g}_{\infty}=[\underline{\mathbf{l}}_{n_g},\overline{\mathbf{l}}_{n_g}]$, where $\underline{\mathbf{l}}_{n_g},\overline{\mathbf{l}}_{n_g}$ are defined below \eqref{eq:und_g_cz}.   
On the other hand, similar to the proof of Lemma \ref{lem:ZB_propagation}, using Corollary \ref{cor:supproting_vectors} we conclude that $\forall H \in \mathbf{H}_{\tilde{f}}$, $\tilde{f}(\cdot)$ can be decomposed as 
\begin{align}\label{eq:decomposed_cz}
\begin{array}{rl}
\tilde{f}(\xi) &=\tilde{g}^H(\xi)+H\xi,  \ \ \ \forall H \in \mathbf{H}_{\tilde{f}},\forall \xi \in \mathbb{I}\tilde{\Xi}, \\
\Rightarrow \tilde{f}(\mathbb{I}\tilde{\Xi}) &\subseteq \tilde{g}^H(\mathbb{I}\tilde{\Xi})\oplus H\mathbb{I}\tilde{\Xi}, \forall H \in \mathbf{H}_{\tilde{f}},
\end{array}
\end{align}
where $\tilde{g}^H(\xi)$ is a JSS function in $\mathbb{I}\tilde{\Xi}$ and $\mathbf{H}_{\tilde{f}}$ is given in \eqref{eq:JSS_maker}. 
Also, by Corollary \ref{cor:supproting_vectors}, 
in each dimension $i \in \mathbb{N}_{n_x}$, $\tilde{g}^H_{i}(\xi)$ can be tightly bounded as $\underline {{g}}^H_{i}\leq \tilde{g}^H_{i}(\xi) \leq \overline {{g}}^H_{i}, \forall \xi \in \mathbb{I}\tilde{\Xi}, \forall H \in \mathbf{H}_{\tilde{f}} $, with $\overline{{g}}^H_{i},\underline{{g}}^H_{i}$ given in \eqref{eq:ov_g_cz} and \eqref{eq:und_g_cz}, respectively. Augmenting all these $\mathbb{N}_{nx}$ one-dimensional inequalities, as well as Proposition \ref{prop:interval_cen_rep} yield the following set inclusion: $\forall H \in \mathbf{H}_{\tilde{f}} $:
\begin{align*}
\tilde{g}^H(\mathbb{I}\tilde{\Xi})\hspace{-.1cm} \subseteq\hspace{-.1cm} [\underline {g}^H,\overline {g}^H]\hspace{-.1cm}=\hspace{-.1cm}\frac{1}{2}((\underline {g}^H+\overline {g}^H)\oplus \textstyle{\mathrm{diag}}(\overline {g}^H-\underline {g}^H)\mathbb{B}^{n_x}_{\infty}).  
\end{align*}
This, \eqref{eq:set_equiv_cz}, \eqref{eq:decomposed_cz} and the fact that the inclusion in \eqref{eq:decomposed_cz} holds for  all $H \in \mathbf{H}_{\tilde{f}}$ and hence for the intersection of all of them, return: $ \forall H \in \mathbf{H}_{\tilde{f}}: f({\mathcal{Z}}) \subseteq \{H\xi+\textstyle{\mathrm{diag}}(\overline {g}^H-\underline {g}^H)\theta+\frac{1}{2}((\underline {g}^H+\overline {g}^H)|\xi \in \mathbb{B}^{n_g}_{\infty},\theta \in \mathbb{B}^{n_x}_{\infty}, \tilde{A}\xi=\tilde{b}\}$, where the set on the right hand side of the inclusion is equivalent to the intersection of the CZs {on} the right hand side of \eqref{eq:CZ_prediction}.
\end{proof}
Finally, for further improvement, we can take the intersection of the resultant propagated sets in Lemmas \ref{lem:ZB_propagation} and \ref{lem:CZ_propagation}.   This, is formally summarized in the following Theorem \ref{thm:ZB_CZ_propagation}. 
\begin{thm}[Decomposition-Based ZB/CZ Propagation]\label{thm:ZB_CZ_propagation}
Suppose all the assumptions in Lemmas \ref{lem:ZB_propagation} and \ref{lem:CZ_propagation} hold. Then, $f(\mathcal{Z}) \subseteq \mathcal{ZB}_f \cap \mathcal{CZ}_f$, where $\mathcal{ZB}_f,\mathcal{CZ}_f$ are computed in Lemmas \ref{lem:ZB_propagation} and \ref{lem:CZ_propagation}, respectively.
\end{thm}
\begin{proof}
It follows from Lemmas \ref{lem:ZB_propagation} and \ref{lem:CZ_propagation} that $f(\mathcal{Z}) \subseteq \mathcal{ZB}_f $ and $f(\mathcal{Z}) \subseteq \mathcal{CZ}_f$, and so $f(\mathcal{Z}) \subseteq \mathcal{ZB}_f \cap \mathcal{CZ}_f$.
\end{proof}
\subsection{Decomposition-Based CZ/ZB update} \label{sec:update}
In this section, we address Problem \ref{prob:update}, given a locally Lipschitz nonlinear vector field $\mu(\cdot)$ and assuming that the initial propagated and the observation/constraint sets are zonotope bundles (Lemma \ref{lem:ZB_update}) or constraint zonotopes (Lemma \ref{lem:CZ_update}). Using similar idea as in Section \ref{sec:propagation}, i.e, lifting to the space of generators, decomposing the transformed observation function into a JSS and a linear component, applying the tight remainder-form decomposition functions \cite{khajenejad2021tightremainder} to bound the JSS component, augmenting and intersecting, as well as taking the advantage of linear remainder functions, we obtain ZB/CZ enclosures to the nonlinear generalized intersection in \eqref{eq:update}. The results of this section are summarized in Lemmas \ref{lem:ZB_update} and \ref{lem:CZ_update} and Theorem \ref{thm:ZB_CZ_update}. 

\begin{lem}[Decomposition-Based ZB Update]\label{lem:ZB_update}
Let $\mu: \mathbb{R}^{n_x}  \to \mathbb{R}^{n_{\mu}}$ satisfies Assumption \ref{assumption:mix-lip}. Let $\mathcal{Z}_f \subset \mathbb{R}^{n_{x}}$ and $Z_{\mu} \subset \mathbb{R}^{n_{\mu}}$ be two ZB sets, i.e., 
$\mathcal{Z}_f=\mathcal{ZB}_f=\bigcap\limits_{r=1}^{R} \{G_f^{r},c_f^{r}\}_Z$ and $\mathcal{Z}_{\mu}=\mathcal{ZB}_{\mu}=\bigcap\limits_{t=1}^{T} \{G_{\mu}^{t},c_{\mu}^{t}\}_Z$,  
and $\forall r \in \mathbb{N}_{R} \triangleq \{1,\dots,R\},\forall t \in \mathbb{N}_{T} \triangleq \{1,\dots,T\}$, $n_r,n_t$ be the number of generators of the corresponding zonotopes, respectively. 
Then, the following set inclusion holds:
\begin{align} \label{eq:ZB_update}
\mathcal{ZB}_f \cap_{\mu} \mathcal{ZB}_{\mu} \hspace{-.1cm}\subseteq \hspace{-.1cm} \mathcal{ZB}_u \hspace{-.1cm} \triangleq \hspace{-.1cm}\bigcap\limits_{r=1}^R\bigcap\limits_{t=1}^T\bigcap\limits_{Q_r \in \mathbf{Q}_{\tilde{\mu}_r}} \{\hat{G}^t_{r},\hat{c}_{r},\hat{A}^{Q_r}_{r,t},\hat{b}^{Q_r}_{r,t}\}_{CZ}, 
\end{align} 
where $ \hat{G}^t_{r} \triangleq [{G}^r_f \ \mathbf{0}^t],\hat{c}_r \triangleq c^r_f, \ \hat{b}^{Q_r}_{r,t} \triangleq c^t_{\mu}-\frac{1}{2}(\overline{p}^{Q_r}_r+\underline{p}^{Q_r}_r)$, 
\begin{align}
\nonumber \hat{A}^{Q_r}_{r,t} &\triangleq \begin{bmatrix} Q_r & -G^t_{\mu} & \frac{1}{2}\textstyle{\mathrm{diag}}(\overline{p}^{Q_r}_r-\underline{p}^{Q_r}_r) \end{bmatrix}, \\
 \overline{p}^{Q_r}_{r,i} &\triangleq  p^{r}_{i,d}(\mathbf{1}_{n_{r}},-\mathbf{1}_{n_{r}};Q^\top_{r(i,:)},\langle \label{eq:over_p_zb} Q^\top_{r(i,:)},\cdot\rangle), \\
 \underline{p}^{Q_r}_{r,i} &\triangleq p^{r}_{i,d}(-\mathbf{1}_{n_{r}},\mathbf{1}_{n_{r}};Q^\top_{r(i,:)},\langle Q^\top_{r(i,:)},\cdot\rangle),\label{eq:under_p_zb}
\end{align}
 $p^{r}_{i,d}(\cdot,\cdot;Q_r,\langle Q^\top_{r(i,:)},\cdot\rangle)$ is the tight mixed-monotone decomposition function (cf. Proposition \ref{prop:Remainder_DF}) for the JSS mapping $p^{Q_r}_{r,i}(\alpha) \triangleq \tilde{\mu}_{r,i}(\alpha)-\langle Q^\top_{r(i,:)},\alpha\rangle:\mathbb{B}^{n_r}_{\infty} \to \mathbb{R}^{n_{\mu}}$, $\mathbf{Q}_{\tilde{\mu}_r}$ is defined similar to $\mathbf{H}_{f}$ in Corollary \ref{cor:supproting_vectors} (with the corresponding function being $\tilde{\mu}_r(\alpha) \triangleq \mu(c^r_{f}+G^r_{f}\alpha)$) and $\mathbf{0}^t$ is a zero matrix in $\mathbb{R}^{n_{x} \times (n_{t}+n_{\mu})}$.
\end{lem}
\begin{proof}
Suppose $z \in \mathcal{ZB}_f \cap_{\mu} \mathcal{ZB}_{\mu}$. Then by definition of the operator $\cap_{\mu}$ (cf. \eqref{eq:update}), $z \in \mathcal{ZB}_f$ and $\mu(z) \in \mathcal{ZB}_{\mu}$. The former implies that $\forall r \in \mathbb{N}_{R}, \exists \alpha \in \mathbb{B}^{n_{r}}_{\infty}$ such that $z=G^r_f\alpha+c^r_f$, while it follows from the latter that $\mu(z)=\mu(G^r_f\alpha+c^r_f) \triangleq \tilde {\mu}_r(\alpha) \in \mathcal{ZB}_{\mu} \Rightarrow \forall t \in \mathbb{N}_T, \exists \zeta \in \mathbb{B}^{n_t}_{\infty}$, such that $\tilde{\mu}_r(\alpha)=c^t_{\mu}+G^t_{\mu}\zeta$.
Putting these two results in a set representation form, we obtain:
\begin{align} \label{eq:int_cz}
z \hspace{-.1cm}\in \hspace{-.1cm} \bigcap\limits_{r=1}^R\hspace{-.1cm}\bigcap\limits_{t=1}^T\{G^r_f\alpha\hspace{-.1cm}+\hspace{-.1cm}c^r_f|\tilde{\mu}_r(\alpha)\hspace{-.1cm}=\hspace{-.1cm}c^t_{\mu}\hspace{-.1cm}+\hspace{-.1cm}G^t_{\mu}\zeta,\alpha \hspace{-.1cm}\in\hspace{-.1cm} \mathbb{B}^{n_{r}}_{\infty},\zeta \hspace{-.1cm}\in\hspace{-.1cm} \mathbb{B}^{n_t}_{\infty} \}   
\end{align}
On the other hand, using Corollary \ref{cor:supproting_vectors}, $\tilde{\mu}_r(\cdot)$ can be decomposed into a JSS and a linear mapping as follows: $\forall r \in \mathbb{N}_R,\forall Q_r \in \mathbf{Q}_{\tilde{\mu}_r}, \forall \alpha \in \mathbb{B}^{n_{r}}_{\infty}$:
\begin{align}\label{eq:dec_zb_update}
\tilde{\mu}_r(\alpha)=p^{Q_r}_r(\alpha)+Q_r\alpha.
\end{align}
Moreover, by the same corollary, the JSS component $p^{Q_r}_r(\cdot)$ is tightly bounded as follows: $\forall i \in \mathbb{N}_{n_{\mu}},\forall Q_r \in \mathbf{Q}_{\tilde{\mu}_r}, \underline{p}^{Q_r}_{r,i} \leq p^{Q_r}_{r.i}(\alpha) \leq  \underline{p}^{Q_r}_{r,i}, \forall \alpha \in \mathbb{B}^{n_{r}}_{\infty}$, with $\underline{p}^{Q_r}_{r,i},\overline{p}^{Q_r}_{r,i}$ given in \eqref{eq:over_p_zb} and \eqref{eq:under_p_zb}, respectively. This, as well as \eqref{eq:dec_zb_update} and Proposition \ref{prop:interval_cen_rep} imply $\forall r \in \mathbb{N}_R,\forall Q_r \in \mathbf{Q}_{\tilde{\mu}_r}, \forall \alpha \in \mathbb{B}^{n_{r}}_{\infty}, \exists \theta \in \mathbb{B}^{n_{\mu}}_{\infty}$ such that $\tilde{\mu}_r(\alpha)=\frac{1}{2}(\underline{p}^{Q_r}_{r,i}+\overline{p}^{Q_r}_{r,i})+\frac{1}{2}\textstyle{\mathrm{diag}}(\overline{p}^{Q_r}_{r,i}-\underline{p}^{Q_r}_{r,i})\theta+Q_r\alpha$, that along with \eqref{eq:int_cz} returns $z \in  \bigcap\limits_{r=1}^R\bigcap\limits_{t=1}^T\bigcap\limits_{Q_t \in \mathbf{Q}_{\tilde{\mu}_t}} \{G^r_f\alpha+c^r_f|\frac{1}{2}(\underline{p}^{Q_r}_{r,i}+\overline{p}^{Q_r}_{r,i})+\frac{1}{2}\textstyle{\mathrm{diag}}(\overline{p}^{Q_r}_{r,i}-\underline{p}^{Q_r}_{r,i})\theta+Q_r\alpha=c^t_{\mu}+G^t_{\mu}\zeta,\alpha \in \mathbb{B}^{n_{r}}_{\infty},\zeta \in \mathbb{B}^{n_t}_{\infty}, \theta \in \mathbb{B}^{n_{\mu}}_{\infty} \}  $, where the set on the right hand side is equivalent to the one on the right hand side of \eqref{eq:ZB_update}.
\end{proof}
\begin{lem}[Decomposition-Based CZ Update]\label{lem:CZ_update}
Let $\mu: \mathbb{R}^{n_x}  \to \mathbb{R}^{n_{\mu}}$ satisfies Assumption \ref{assumption:mix-lip}. Let $\mathcal{Z}_f \subset \mathbb{R}^{n_{x}}$ and $\mathcal{Z}_{\mu} \subset \mathbb{R}^{n_{\mu}}$ be two CZ sets, i.e., 
$\mathcal{Z}_f=\mathcal{CZ}_f= \{\tilde{G}_f,\tilde{c}_f,\tilde{A}_f,\tilde{b}_f\}_{CZ}$ and $\mathcal{Z}_{\mu}=\mathcal{CZ}_{\mu}= \{\tilde{G}_{\mu},\tilde{c}_{\mu},\tilde{A}_{\mu},\tilde{b}_{\mu}\}_{CZ}$,  
and $n_c,n_{\tau}$ be the number of generators of $\mathcal{Z}_f,\mathcal{Z}_{\mu}$, respectively. Then, the following set inclusion holds:
\begin{align} \label{eq:CZ_update}
\mathcal{CZ}_f \cap_{\mu} \mathcal{CZ}_{\mu} \subseteq \mathcal{CZ}_u \triangleq \bigcap\limits_{\Omega \in \mathbf{\Omega}_{{\lambda}}} \{\mathbb{G},\tilde{c}_{f},\mathbb{A}_{\Omega},\tilde{b}_{\Omega}\}_{CZ}, 
\end{align} 
where $\mathbb{G} \triangleq [\tilde{G}_f \ {0} \ 0], \ \tilde{b}_{\Omega} \triangleq [\tilde{b}^\top_f \ \tilde{b}^\top_{\mu} \ (\tilde{c}_f-\frac{1}{2}(\overline{\nu}^{\Omega}+\underline{\nu}^{\Omega}))^\top]^\top$,
\begin{align}
\nonumber \mathbb{A}_{\Omega} &\triangleq \begin{bmatrix} \tilde{A}_f & 0 & 0 \\ 0 & \tilde{A}_{\mu} & \\ \Omega & -\tilde{G}_{\mu} & \frac{1}{2}\textstyle{\mathrm{diag}}(\overline{\nu}^{\Omega}-\underline{\nu}^{\Omega}) \end{bmatrix}, \\
 \overline{\nu}^{\Omega}_{i} &\triangleq  {\nu}_{i,d}(\overline{\mathbf{l}}_{n_{c}},\underline{\mathbf{l}}_{n_{c}};\Omega^\top_{(i,:)},\langle \label{eq:over_p_cz} \Omega^\top_{(i,:)},\cdot\rangle), \\
 \underline{\nu}^{\Omega}_{i} &\triangleq {\nu}_{i,d}(\underline{\mathbf{l}}_{n_{c}},\overline{\mathbf{l}}_{n_{c}};\Omega^\top_{(i,:)},\langle \Omega^\top_{(i,:)},\cdot\rangle),\label{eq:under_p_cz}\\
 \nonumber \overline{\mathbf{l}}_{n_{c}} & \triangleq \hspace{-.1cm} \min(\mathbf{1}_{n_c},\tilde{A}_f^\dagger\tilde{b}_f+\kappa \mathbf{r}_{n_c}),\underline{\mathbf{l}}_{n_{c}} \hspace{-.15cm} \triangleq \hspace{-.1cm} \max(-\mathbf{1}_{n_c},\tilde{A}_f^\dagger\tilde{b}_f\hspace{-.1cm}-\hspace{-.1cm}\kappa \mathbf{r}_{n_c}\hspace{-.05cm})
\end{align}
 ${\nu}_{i,d}(\cdot,\cdot;{\Omega}^\top_{(i,:)},\langle \Omega^\top_{(i,:)},\cdot\rangle)$ is the tight mixed-monotone decomposition function (cf. Proposition \ref{prop:Remainder_DF}) for the JSS mapping $\nu^{\Omega}_{i}(\beta) \triangleq {\lambda}_{i}(\beta)-\langle \Omega^\top_{(i,:)},\beta\rangle:\mathbb{B}^{n_c}_{\infty} \to \mathbb{R}^{n_{\mu}}$, $\mathbf{\Omega}_{{\lambda}}$ is defined similar to $\mathbf{H}_{f}$ in Corollary \ref{cor:supproting_vectors} (with the corresponding function being ${\lambda}(\beta) \triangleq \mu(\tilde{c}_{f}+\tilde{G}_{f}\beta)$), $\mathbf{r}_{n_c} \triangleq \textstyle{\mathrm{rowsupp}}(I_{n_c}-\tilde{A}_f^\dagger \tilde{A}_f)$ and $\kappa$ is a very large positive real number (infinity).
\end{lem}
\begin{proof}
Suppose $z \in \mathcal{CZ}_f \cap_{\mu} \mathcal{CZ}_{\mu}$. Then by definition of the operator $\cap_{\mu}$ (cf. \eqref{eq:update}), $z \in \mathcal{CZ}_f$ and $\mu(z) \in \mathcal{CZ}_{\mu}$. The former implies that $ \exists \beta \in \mathbb{B}^{n_{c}}_{\infty}$ such that $\tilde{A}_f \beta=\tilde{b}_f \land z=\tilde{G}_f\beta+\tilde{c}_f$, while it follows from the latter that $\mu(z)=\mu(\tilde{G}_f\beta+\tilde{c}_f) \triangleq  {\lambda}(\beta) \in \mathcal{CZ}_{\mu} \Rightarrow , \exists \gamma \in \mathbb{B}^{n_{\tau}}_{\infty}$, such that $\tilde{A}_{\mu}\gamma=\tilde{b}_{\mu} \land {\lambda}(\beta)=\tilde{c}_{\mu}+\tilde{G}_{\mu}\gamma$.
Putting these two results into a set representation form, we obtain:
{\small
\begin{align} \label{eq:int_cz2}
z \hspace{-.1cm}\in \hspace{-.1cm} \{\tilde{G}_f\beta\hspace{-.1cm}+\hspace{-.1cm}\tilde{c}_f|\lambda(\beta)\hspace{-.1cm}=\hspace{-.1cm}\tilde{c}_{\mu}\hspace{-.1cm}+\hspace{-.1cm}\tilde{G}_{\mu}\gamma,\hspace{-.1cm} \tilde{A}_f \beta\hspace{-.1cm}=\hspace{-.1cm}\tilde{b}_f,\tilde{A}_{\mu}\gamma\hspace{-.1cm}=\hspace{-.1cm}\tilde{b}_{\mu},
\beta \hspace{-.1cm}\in\hspace{-.1cm} \mathbb{B}^{n_{c}}_{\infty},\gamma \hspace{-.1cm}\in\hspace{-.1cm} \mathbb{B}^{n_{\tau}}_{\infty} \} 
\end{align}
}
On the other hand, using Corollary \ref{cor:supproting_vectors}, $\lambda(\cdot)$ can be decomposed into a JSS and a linear mapping as follows: 
\begin{align}\label{eq:dec_cz_update}
\forall \Omega \in \mathbf{\Omega}_{\lambda}, \forall \beta \in \mathbb{B}^{n_{c}}_{\infty}:\lambda(\beta)={\nu}^{\Omega}(\beta)+\Omega\beta.
\end{align}
Further, note that by \cite[Theorem 2]{james1978generalised}, $\tilde{A}_f\beta=\tilde{b}_f \Rightarrow \beta \in \mathbb{IB} \triangleq [\tilde{A}_f^\dagger\tilde{b}_f-\kappa \mathbf{r}_{n_c},\tilde{A}_f^\dagger\tilde{b}_f-\kappa \mathbf{r}_{n_c}]$, where $\mathbf{r}_{n_c} \triangleq \textstyle{\mathrm{rowsupp}}(I_{n_c}-\tilde{A}_f^\dagger \tilde{A}_f)$ and $\kappa$ is a very large positive real number, which given that $\beta \in \mathbb{B}^{n_c}_{\infty}$, results is $\beta \in \mathbb{IB} \cap \mathbb{B}^{n_c}_{\infty}=[\underline{\mathbf{l}}_{n_{c}},\overline{\mathbf{l}}_{n_{c}}]$. This and Corollary \ref{cor:supproting_vectors} imply that the JSS component ${\nu}^{\Omega}(\cdot)$ is tightly bounded as follows: $\forall i \in \mathbb{N}_{n_{\mu}},\forall \Omega \in \mathbf{\Omega}_{\lambda}, \underline{\nu}^{\Omega}_{i} \leq {\nu}^{\Omega}_{i}(\beta) \leq  \underline{\nu}^{\Omega}_{i}, \forall \beta \in \mathbb{B}^{n_{c}}_{\infty}$, with $\underline{\nu}^{\Omega}_{i},\overline{\nu}^{\Omega}_{i}$ given in \eqref{eq:over_p_cz} and \eqref{eq:under_p_cz}, respectively. This, \eqref{eq:dec_cz_update} and Proposition \ref{prop:interval_cen_rep} imply $\forall \Omega \in \mathbf{\Omega}_{\lambda}, \forall \beta \in \mathbb{B}^{n_{c}}_{\infty}, \exists \rho \in \mathbb{B}^{n_{\mu}}_{\infty}$ such that $\lambda(\beta)=\frac{1}{2}(\underline{\nu}^{\Omega}+\overline{\nu}^{\Omega})+\frac{1}{2}\textstyle{\mathrm{diag}}(\overline{\nu}^{\Omega}-\underline{\nu}^{\Omega})\rho+\Omega\alpha$, which along with \eqref{eq:int_cz2} returns $z \in  \bigcap\limits_{\Omega \in \mathbf{\Omega}_{\lambda}} \{\tilde{G}_f\beta+\tilde{c}_f|\frac{1}{2}(\underline{\nu}^{\Omega}+\overline{\nu}^{\Omega})+\frac{1}{2}\textstyle{\mathrm{diag}}(\overline{\nu}^{\Omega}-\underline{\nu}^{\Omega})\rho+\Omega\beta=\tilde{c}_{\mu}+\tilde{G}_{\mu}\gamma,\tilde{A}_f\beta=\tilde{b}_f,\tilde{A}_{\mu}\gamma=\tilde{b}_{\mu},\beta \in \mathbb{B}^{n_{c}}_{\infty},\gamma \in \mathbb{B}^{n_{\tau}}_{\infty}, \rho \in \mathbb{B}^{n_{\mu}}_{\infty} \}  $, where the set on the right hand side is equivalent to the one on the right hand side of \eqref{eq:ZB_update}. 
\end{proof}\vspace{-.2cm}
We conclude this subsection by combining the results in Lemmas \ref{lem:ZB_update} and \ref{lem:CZ_update} via the following Theorem \ref{thm:ZB_CZ_update}. 
\begin{thm}[Decomposition-Based ZB/CZ Update]\label{thm:ZB_CZ_update}
Suppose all the assumptions in Lemmas \ref{lem:ZB_update} and \ref{lem:CZ_update} hold. Then $$\mathcal{Z}_f \cap_{\mu} \mathcal{Z}_{\mu} \subseteq \mathcal{ZB}_u \cap \mathcal{CZ}_u,$$ where $\mathcal{ZB}_u,\mathcal{CZ}_u$ are given in Lemmas \ref{lem:ZB_update} and \ref{lem:CZ_update}, respectively.
\end{thm}
\begin{proof}
By Lemmas \ref{lem:ZB_update} and \ref{lem:CZ_update}: $\mathcal{Z}_f \cap_{\mu} \mathcal{Z}_{\mu} \subseteq \mathcal{ZB}_u $ and $\mathcal{Z}_f \cap_{\mu} \mathcal{Z}_{\mu} \subseteq \mathcal{CZ}_u$, and hence $\mathcal{Z}_f \cap_{\mu} \mathcal{Z}_{\mu} \subseteq \mathcal{ZB}_u \cap \mathcal{CZ}_u$.
\end{proof}
\subsection{Modifications to The Approach in \cite{rego2020guaranteed}}
The purpose of this subsection is twofold. i) We make a potential refinement/improvement to the propagation approach in \cite[Theorem 2]{rego2020guaranteed} (recapped in Proposition \ref{prop:Rego_propagation}) through the following Proposition \ref{prop:Rego_propagation_refinement}, by applying our previously developed remainder-form decomposition functions to compute potentially tighter enclosing intervals to Jacobian matrix of $f(\cdot)$. ii) We propose an update method via Lemma \ref{lem:Rego_update}, that is based on the ``CZ-inclusion" introduced in \cite[Theorem 1]{rego2020guaranteed} (recapped in Proposition \ref{prop:Rego_CZ_bounding}). The proposed update method is applicable to general nonlinear observation functions (similar to the proposed methods in Lemmas \ref{lem:ZB_update} and \ref{lem:CZ_update}), as opposed to the update (i.e, linear intersection) approach in \cite{rego2020guaranteed} that is only applicable when the observation function is linear. 
\begin{prop}[Refinement to The Propagation Approach in \cite {rego2020guaranteed}]\label{prop:Rego_propagation_refinement}
Suppose all the assumptions in Proposition \ref{prop:Remainder_DF} (i.e, \cite[Theorem 2]{rego2020guaranteed}) hold. Then the set inclusion in \eqref{eq:Rego_prediction} holds also with replacing $\mathbb{J}$ with $\tilde{\mathbb{J}}$ (or the best (tightest) of them), where $\tilde{\mathbb{J}}$ is an enclosing interval to $g(x)\triangleq \nabla^\top_x f(X,W)$ that can be computed by applying Proposition \ref{prop:Remainder_DF} to the function $g(\cdot)$.
\end{prop}
\begin{proof}
Directly follows from Proposition \ref{prop:Remainder_DF}.
\end{proof}
\begin{lem}[Update Based on ``CZ-Inclusion" in \cite {rego2020guaranteed}]\label{lem:Rego_update}
Suppose all the assumptions in Lemma \ref{lem:CZ_update} hold. Let $x_0 \in \mathcal{CZ}_f$ and $\mathbb{J}^{\mu},\mathbb{J}^{\mu}_{\Delta} \in \mathbb{R}^{n_{\mu} \times n_x}$ be interval matrices satisfying $J^\mu(\mathcal{CZ}_f) \subseteq \mathbb{J}^{\mu}$ and $\forall i \in \mathbb{N}_{n_{\mu}},\forall j \in \mathbb{N}_{n_{x}},[\mathbb{J}^{\mu}_{\Delta}]_{ij} \triangleq \frac{1}{2}\begin{bmatrix} - \textstyle{\mathrm{diam}}(\mathbb{J}^{\mu})_{ij} & \textstyle{\mathrm{diam}}(\mathbb{J}^{\mu})_{ij} \end{bmatrix}$, where $J^{\mu}$ denotes the Jacobian of $\mu(\cdot)$. Let $\overline{\mathcal{Z}}_{f}=\{\overline{G}^{f},\overline{c}^{f}\}_Z$ be a zonotope satisfying $\mathcal{CZ}_f \ominus x_0 \subseteq \overline{\mathcal{Z}}_f$, with $\overline{c}^f \in \mathbb{R}^{\overline{n}}$. Let $\mathbf{m}^{\mu} \in \mathbb{R}^{n_{\mu}}$ be an interval vector such that $\mathbf{m}^{\mu} \supset \mathbb{J}_{\Delta}^{\mu}\overline{c}^f$ and $\textstyle{\mathrm{mid}}(\mathbf{m}^{\mu})=\mathbf{0}_{n_{\mu}}$. Let $P^{\mu} \in \mathbb{R}^{n_{\mu} \times n_{\mu}}$ be a diagonal matrix defined as follows: $\forall i=1,\dots,n_{\mu}$: 
\begin{align}\label{eq:Rego_P_update}
P^{\mu}_{ii}\hspace{-.1cm}=\hspace{-.1cm}\frac{1}{2}\textstyle{\mathrm{diam}}(\mathbf{m}^{\mu})_i\hspace{-.1cm}+\hspace{-.1cm}\frac{1}{2}\sum_{j=1}^{\overline{n}} \sum_{k=1}^{n_x} \textstyle{\mathrm{diam}}(\mathbb{J}^{\mu}_{\Delta})_{ik}|\overline{G}^f_{kj}|.
\end{align}
Then, the following set inclusion holds:
\begin{align} \label{eq:Rego_update}
\mathcal{CZ}_f \cap_{\mu} \mathcal{CZ}_{\mu} \subseteq \mathcal{CZ}^R_u \triangleq \{{G}_u,{c}_u,{A}_{u},{b}_{u}\}_{CZ}, 
\end{align} 
where $ {b}_{u}\hspace{-.1cm} \triangleq\hspace{-.1cm} [(\tilde{c}_{\mu}\hspace{-.1cm}-\mu(x_0)+\textstyle{\mathrm{mid}}(\mathbb{J}^{\mu})(x_0-\tilde{c}_f)-c_R)^\top \ \tilde{b}^\top_{f} \ \tilde{b}^\top_{\mu} \ {b}^\top_{R} ]^\top$
\begin{align}
\nonumber {A}_{u} &\triangleq \begin{bmatrix} \textstyle{\mathrm{mid}}(\mathbb{J}^{\mu})\tilde{G}_f & -\tilde{G}_{\mu} & G_R \\ \tilde{A}_f & 0 & 0 \\ 0 & \tilde{A}_{\mu} & 0 \\ 0 & 0 & A_R \end{bmatrix}, \ {G}_u \triangleq [\tilde{G}_f \ {0} \ 0], \\
  G_R &\triangleq [0 \ P^{\mu}], \ c_R \triangleq 0, \ A_R \triangleq [\tilde{A}_f \ 0], \ b_R \triangleq \tilde{b}_f.\label{eq:Rego_param}
\end{align} 
\end{lem}
\begin{proof}
Suppose $z \in \mathcal{CZ}_f \cap_{\mu} \mathcal{CZ}_{\mu}$. Then by definition of the operator $\cap_{\mu}$ (cf. \eqref{eq:update}), $z \in \mathcal{CZ}_f$ and $\mu(z) \in \mathcal{CZ}_{\mu}$. Note that by Proposition \ref{prop:interval_cen_rep} and the mean value theorem, $z \in \mathcal{CZ}_f \Rightarrow \mu(z) \in \mu(\mathcal{CZ}_f) \subseteq \mu(x_0)\oplus \mathbb{J}^{\mu}(\mathcal{CZ}_f \ominus x_0)=\mu(x_0)\oplus (\textstyle{\mathrm{mid}}(\mathbb{J}^{\mu})+\mathbb{J}^{\mu}_{\Delta})(\mathcal{CZ}_f \ominus x_0)=$
\begin{align}\label{eq:MVT}
\mu(x_0)-\textstyle{\mathrm{mid}}(\mathbb{J}^{\mu})\oplus \textstyle{\mathrm{mid}}(\mathbb{J}^{\mu})\mathcal{CZ}_f\oplus \mathbb{J}^{\mu}_{\Delta}(\mathcal{CZ}_f \ominus x_0).
\end{align}
On the other hand, by Proposition \ref{prop:Rego_CZ_bounding} : 
\begin{align}\label{Rego_inclusion}
\mathbb{J}^{\mu}_{\Delta}(\mathcal{CZ}_f \ominus x_0) \subseteq \mathcal{CZ}_R \triangleq \{G_R,c_R,A_R,b_R\}_{CZ},
\end{align}
with $G_R,c_R,A_R,b_R$ given in \eqref{eq:Rego_P_update} and \eqref{eq:Rego_param} (note that $\textstyle{\mathrm{mid}}(\mathbb{J}^{\mu}_{\Delta})=0$ by its definition) and where $\mathcal{CZ}_R$ has $n_R$ generators. Then, the facts that $z \in \mathcal{CZ}_f \triangleq \{\tilde{G}_f\beta+\tilde{c}_f |\tilde{A}_f \beta=\tilde{b}_f,\beta\in \mathbb{B}^{n_{c}}_{\infty}\}$, $\mu(z) \in \mathcal{CZ}_{\mu} \triangleq \{\tilde{G}_{\mu}\gamma+\tilde{c}_{\mu} |\tilde{A}_{\mu} \gamma=\tilde{b}_{\mu},\gamma \in \mathbb{B}^{n_{\tau}}_{\infty}\}$, \eqref{eq:MVT} and \eqref{Rego_inclusion} imply that
$z \in  \{\tilde{G}_f\beta+\tilde{c}_f|\tilde{c}_{\mu}+\tilde{G}_{\mu}\gamma=\mu(x_0)+\textstyle{\mathrm{mid}}(\mathbb{J}^{\mu})(\tilde{c}_f-x_0)+\textstyle{\mathrm{mid}}(\mathbb{J}^{\mu})\beta+C_R+G_R\xi_R, \tilde{A}_f \beta=\tilde{b}_f,\tilde{A}_{\mu}\gamma=\tilde{b}_{\mu},A_R \xi_R=b_R,
\beta\in \mathbb{B}^{n_{c}}_{\infty},\gamma \in \mathbb{B}^{n_{\tau}}_{\infty},\xi_R \in \mathbb{B}^{n_{R}}_{\infty}\}$, where the set on the right hand side is equivalent to the CZ on the right hand side of \eqref{eq:Rego_update}. 
\end{proof}

 \section{Simulations} \label{sec:examples}\vspace{-0.05cm}
In this section we compare the performance of five approaches to guaranteed state estimation: i) {RRSR}, i.e., the mean value extension-based propagation introduced in \cite{rego2020guaranteed} (recapped in Proposition \ref{prop:Rego_propagation}) in addition to the update approach in \cite{rego2020guaranteed} for the case when the observation function is linear (e.g., Example I below) and its extension in Lemma \ref{lem:Rego_update} to nonlinear measurements (e.g., Example II below), ii) {D-RRSR}, i.e, a modification to RRSR where the bounds for Jacobian matrices are computed using the reminder-form decomposition functions (cf. Proposition \ref{prop:Rego_propagation_refinement}), iii) {D-ZB}, i.e., decomposition-based propagation and update with ZBs (cf. Lemmas \ref{lem:ZB_propagation} and \ref{lem:ZB_update}), iv) {D-CZ}, i.e., decomposition-based propagation and update with CZs (cf. Lemmas \ref{lem:CZ_propagation} and \ref{lem:CZ_update}) and v) {COMB}, i.e., a combination of i)--v) via intersection.   
 \subsection{Example I}
Consider the following discrete-time system \cite[Section]{rego2020guaranteed}
\begin{align*}
\begin{array}{rll}
x_{1,k}&=3x_{1,k-1}-\frac{x_{1,k-1}^2}{7}-\frac{4x_{1,k-1}x_{2,k-1}}{4+x_{1,k-1}}+w_{1,k-1},\\
x_{2,k}&=-2x_{2,k-1}+\frac{3x_{1,k-1}x_{2,k-1}}{4+x_{1,k-1}}+w_{2,k-1},\\
\begin{bmatrix}
y_{1,k}\\
y_{2,k}
\end{bmatrix}&=
\begin{bmatrix}
1 & 0\\
-1 & 1\end{bmatrix}
\begin{bmatrix}
x_{1,k}\\
x_{2,k}\end{bmatrix}+
\begin{bmatrix}
v_{1,k}\\
v_{2,k}
\end{bmatrix},
\end{array}
\end{align*}
with $\left\lVert w_k\right\rVert_\infty \leq 0.1$, $\left\lVert v_k\right\rVert_\infty \leq 0.4$, 
and an initial zonotopic enclose for the initial state:
    $ \mathcal{X}_0=\{\begin{bmatrix}
     0.1& 0.2& -0.1\\
     0.1&0.1&0
     \end{bmatrix},
     \begin{bmatrix}
     0.5\\0.5
     \end{bmatrix}\}$.

\begin{figure}[h]
\includegraphics[scale=0.152,trim=27mm 0mm 5mm 5mm,clip]{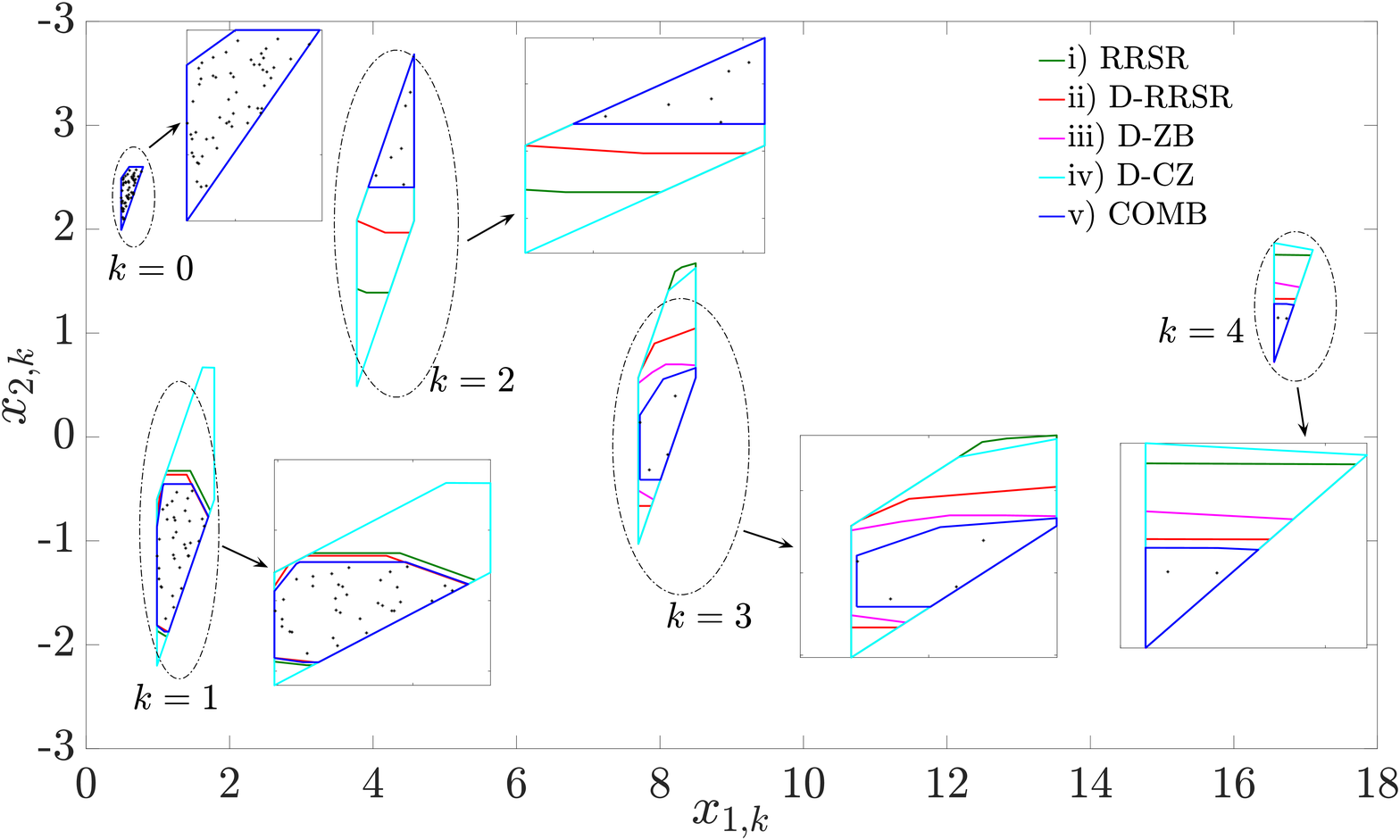}
\caption{Results from the first five time steps of set-valued state estimation, using five different approaches, noticeable from the label on the figure. Black dots are obtained from uniform sampling of $X_0$ propagated through the vector field in Example I\label{Ex_Rego}}
\end{figure}
As can be seen from Figure \ref{Ex_Rego}, {D-ZB} provides less conservative enclosures compared to the other individual approaches, and further, the {COMB} approach results in significant improvement, taking the advantage of intersection. Moreover, a comparison between corresponding average times and the enclosure set volumes to the five approaches shown in Table \ref{Table_Rego}. As it can be observed, D-ZB gives the fastest responses, while combination of all, as expected, took the longest, while RRSR and D-RRSR took approximately the same time on average. In terms of volume, D-ZB and D-RRSR generate the least conservative enclosures compared to the other individual approaches, while a further improvement is obtained using the the intersection of all the individual approaches, i.e., COMB.
\begin{table}[H]
\captionsetup{font=small}
\caption{Average total times (second) and average total volumes ($10^{-4}$) at each time step of the estimators is calculated. Each average is taken over $50$ separate simulations.}

\begin{centering}
\tabcolsep=0.11cm 
\begin{tabular}{|l|l l l l l l|}

\hline
\multicolumn{2}{|c}{Methods:\quad  \quad \quad \quad } & $k=0$ & $k=1$ & $k=2$ & $k=3$ &$k=4$\tabularnewline
\hline 
\multirow{2}{*}{RRSR} & Time: & $0.087$ & $0.250$ & $0.193$ & $0.196$ & $0.204$\\
 & Vol.: & $0.5196$ & $2.3522$ & $3.6224$ & $3.0243$ & $1.8968$\tabularnewline
\hline 
\multirow{2}{*}{D-RRSR} & Time: & $0.022$ & $0.225$ & $0.181$ & $0.198$ & $0.201$\tabularnewline
 & Vol.: & $0.5196$ & $2.1062$ & $3.2979$ & $1.3811$ & $0.6320$\tabularnewline
\hline 
\multirow{2}{*}{D-ZB} & Time: & $0.192$ & $0.095$ & $0.091$ & $0.091$ & $0.123$\tabularnewline
 & Vol.: & $0.5196$ & $0.7974$ & $0.9918$ & $0.8316$ & $0.6170$\tabularnewline
\hline 
\multirow{2}{*}{D-CZ} & Time: & $0.065$ & $2.825$ & $2.920$ & $2.118$ & $3.318$\tabularnewline
 & Vol.: & $0.5196$ & $1.8857$ & $2.1534$ & $1.7924$ & $1.3764$\tabularnewline
\hline 
\multirow{2}{*}{COMB} & Time: & $0.063$ & $6.193$ & $6.882$ & $6.278$ & $6.951$\tabularnewline
 & Vol.: & $0.5196$ & $1.4226$ & $1.8399$ & $1.0038$ & $0.6387$\tabularnewline
\hline 
\end{tabular}
\label{Table_Rego}
\par\end{centering}
\end{table}
  \subsection{Example II (The Unicycle System)}
 Now consider the following discretized unicycle-like mobile robot system \cite{chen2018nonlinear} in the form of \eqref{eq:system}: 
  \begin{align}\label{eq:unicyclemodel}
  \begin{array}{rlll}
      s_{x,k+1}&=s_{x,k}+T_0\phi_{w}\cos(\theta_k)+w_{1,k},\\
       s_{y,k+1}&=s_{y,k}+T_0\phi_{w}\sin (\theta_{k})+w_{2,k},\\
      \theta_{k+1}&=\theta_k+T_0\phi_{\theta}+w_{3,k},\\
      y_k&=[d_{1,k} \ \phi_{1,k} \ d_{2,k} \ \phi_{2,k}]^\top +v_k,
      \end{array}
  \end{align}
where $x_k \triangleq [s_{x,k} \ s_{y,k} \ \theta_k]^\top$, $w_k=[w_{x,k} \ w_{y,k} \ w_{\theta,k}]^\top, \phi_{\omega,k}=0.3, \phi_{\theta,k}=0.15,  w_{x,k}=0.2(0.5\rho_{x_{1,k}}-0.3), w_{y,k}=0.2(0.3\rho_{x_{2,k}}-0.2)$ and $w_{\theta,k}=0.2(0.6\rho_{x_{3,k}}-0.4)$, with $\rho_{x_{l,k}} \in [0,1]$ $(l=1,2,3)$ and initial state $x_0=[0.1 \ 0.2 \ 1]^\top$. Moreover, $\forall i \in \{1,2\}$, $d_{i,k}=\sqrt{(s_{x_i}-s_{x,k})^2+(s_{y_i}-s_{y,k})^2}$ and $\phi_{i,k}=\theta_k-\arctan(\frac{s_{y_i}-s_{y,k}}{s_{x_i}-s_{x,k}})$, with $s_{x_i} , s_{y_i}$ being two known values. Furthermore, $v_{1,k}=0.02 \rho_{{y_1},k}-0.01$, $v_{2,k}=0.03 \rho_{{y_2},k}-0.01$, $v_{3,k}=0.03 \rho_{{y_3},k}-0.02$, $v_{4,k}=0.05 \rho_{{y_4},k}-0.03$ and $\rho_{{y_k},k} \in [0,1]$ $(k=1,2,3,4)$. Applying all the methods, one can observe from Figure \ref{Ex_Unicycle}
that the resultant set estimates are very comparable for all the different five approaches. In terms of computation time, Table \ref{Table_Unicycle} shows that D-CZ takes minimum computation time followed by RRSR, D-RRSR, COMB and D-ZB. In terms of set volumes, the COMB approach takes minimum volume followed by D-ZB, D-CZ, RRSR and D-RRSR. Note that the computation time for D-ZB's is exceptionally large, presumably because the conversion of a polytope to a zonotope bundle in CORA usually results in a higher number of zonotopes than the needed minimal  number of zonotopes in the zonotope bundle. 
The reduction of this number of zonotopes in the bundle could be an interesting topic to explore in the future, which could significantly decrease the computation time of our D-ZB approach. 
  \begin{figure}[h]
\includegraphics[scale=0.152,trim=20mm 0mm 5mm 5mm,clip]{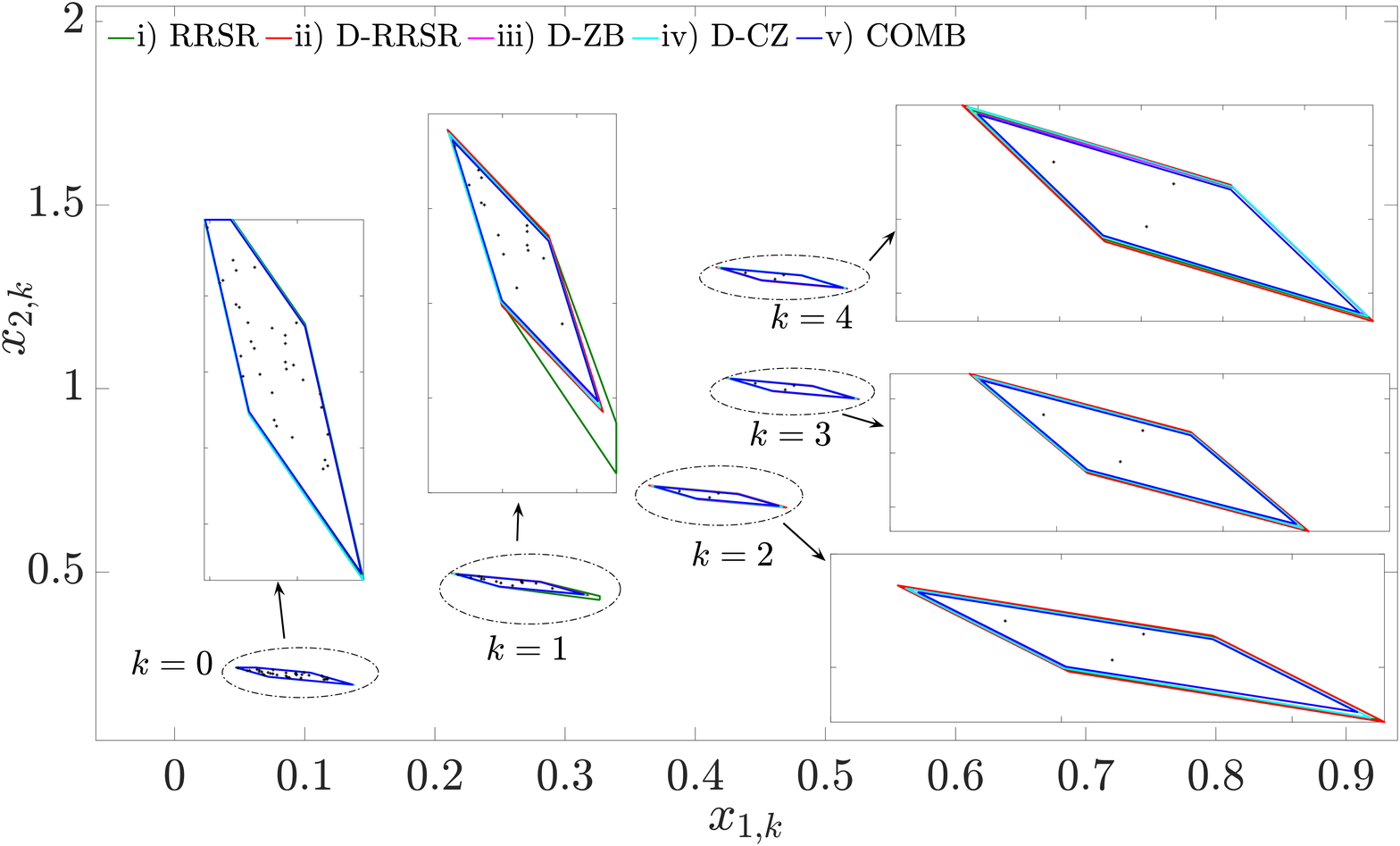}\vspace{-.15cm}
\caption{Results from the first five time steps of set-valued state estimation, using five different approaches, noticeable from the label on the figure. Black dots are obtained from uniform sampling of $X_0$ propagated via \eqref{eq:unicyclemodel}}\label{Ex_Unicycle}
\end{figure}
  

\begin{table}
\captionsetup{font=small}
\caption{Average total times (second) and average total volumes ($10^{-4}$) at each time step of the estimators is calculated. Each average is taken over $20$ separate simulations.}

\centering{}%
\tabcolsep=0.11cm 
\begin{tabular}{|l| l l l l l l|}
\hline 
\multicolumn{2}{|c}{Methods:  \quad \quad \quad } & $k=0$ & $k=1$ &$k=2$& $k=3$ & $k=4$\tabularnewline
\hline 
\multirow{2}{*}{RRSR} & Time: & $0.772$ & $4.256$ & $4.188$ & $2.995$ & $3.675$\tabularnewline
 & Vol.: & $1.7309$ & $1.0293$ & $5.528$ & $1.377$ & $1.413$\tabularnewline
\hline 
\multirow{2}{*}{D-RRSR} & Time: & $1.669$ & $42.90$ & $45.57$ & $28.64$ & $50.53$\tabularnewline
 & Vol.: & $1.279$ & $9.717$ & $5.654$ & $1.978$ & $1.888$\tabularnewline
\hline 
\multirow{2}{*}{D-ZB} & Time: & $1.397$ & $34.20$ & $163.08$ & $147.75$ & $131.75$\tabularnewline
 & Vol.: & $0.042$ & $0.030$ & $0.013$ & $0.043$ & $0.042$\tabularnewline
\hline 
\multirow{2}{*}{D-CZ} & Time: & $0.602$ & $2.182$ & $2.016$ & $2.228$ & $2.591$\tabularnewline
 & Vol.: & $0.073$ & $0.072$ & $0.030$ & $1.261$ & $0.091$\tabularnewline
\hline 
\multirow{2}{*}{COMB} & Time: & $0.236$ & $34.90$ & $65.50$ & $62.37$ & $57.73$\tabularnewline
 & Vol.: & $0.041$ & $0.028$ & $0.012$ & $0.041$ & $0.040$\tabularnewline
\hline 
\end{tabular}
\label{Table_Unicycle}
\end{table}
\vspace{-.4cm}
\section{Conclusion} \label{sec:conclusion}
New set-membership methods were presented in this paper for   state estimation in   bounded-error discrete-time nonlinear systems, subject to   nonlinear observations/constraints. By transforming ZB/CZ uncertainty sets from the state space  to  the space of ZB/CZ's generators, our recently developed tight remainder-form mixed-monotone decomposition functions were applied to compute enclosures that were guaranteed to enclose the set of the state trajectories of the system. Further, remainder-form decomposition functions were leveraged to bound the nonlinear observation function to derive the updated set, i.e., to return enclosures to intersection of the propagated set and the set of states that are consistent with the measurements. Finally, the mean value extension-based approach in \cite{rego2020guaranteed} was also generalized to compute the updated set for nonlinear observations. 

{\tiny
\bibliographystyle{unsrturl}
\bibliography{biblio}
}
\end{document}